\documentclass[11pt]{article}

\usepackage[utf8]{inputenc}
\usepackage{amsfonts}
\usepackage{amssymb}
\usepackage[section]{placeins}
\usepackage{amsthm}
\usepackage{latexsym,amsmath}

\usepackage{graphicx}
\usepackage{microtype}
\usepackage[
    top=2.5cm,
    bottom=2.5cm,
    left=2.5cm,
    right=2.5cm
]{geometry}
\usepackage{threeparttable}
\usepackage[authoryear, round]{natbib}
\setcitestyle{aysep={, }}
\usepackage{placeins}
\usepackage{adjustbox}
\usepackage{booktabs}
\usepackage{multirow}
\usepackage{upgreek}
\usepackage{algorithm}
\usepackage{algpseudocode}
\usepackage{xcolor}
\usepackage{float}
\usepackage{bm}
\usepackage{makecell}
\usepackage{array}
\usepackage{authblk}

\newtheorem{theorem}{Theorem}
\newtheorem{lemma}{Lemma}

\newtheorem{prop}{Proposition}

\newtheorem{definition}{Definition}

\newtheorem{assump}{Assumption}

\usepackage[colorlinks,linkcolor=blue,citecolor=blue,urlcolor=blue]{hyperref}

\title{DASH: A Dimensionality Reduction Method for Large-scale Convex MIQP with Applications in Subset Portfolio Selection}

\author[1]{Pinzhang (Jack) Cheng\thanks{Email: \texttt{pinzhang.cheng@gmail.com}}}
\affil[1]{School of Mathematics and Statistics, UNSW Sydney}

\date{\today}

\begin{document}

\maketitle

\begin{abstract}
Subset selection problems as MIPs (Mixed Integer Programs) are NP-hard. For large scale problems, it is infeasible to find global optimal solutions in a reasonable time and good-quality incumbent solutions are sought after with MIP solvers in practice. This paper proposes DASH (Decreasing Active Set Hierarchy) -- a dimensionality reduction method that improves the MIP solver performance for a subclass of best subset selection problems that can be formulated as MIQPs (Mixed Integer Quadratic Programs). We develop and evaluate the performance of DASH in the subset portfolio selection problem with comparison to Gurobi, a commercial MIP solver. In addition to the problem size, the difficulty of a problem is related to the condition number of the covariance matrix and the box constraint on portfolio weights. An extensive set of numerical experiments with varying problem configurations shows that DASH offers consistent and significant improvement of incumbent solutions when the problem is difficult to solve by Gurobi. In particular, the magnitude and duration of improvement by DASH scale with the difficulty of the problem.
\end{abstract}

\section{Introduction}
In this paper, we propose DASH (Decreasing Active Set Hierarchy), a dimensionality reduction
method that improves the early-stage incumbent solutions generated by MIP solvers for a subclass
of subset selection problems that can be formulated as large-scale convex MIQPs (Mixed Integer
Quadratic Programs). The goal of a subset selection problem is to select an optimal subset, measured by an objective function, from all subsets that satisfies a given sparsity level. For a full set of size $p$ and a target sparsity level of $k$, there are $\binom{p}{k}$ feasible subsets. Because of this combinatorial nature, the number of feasible subsets (points) grows exponentially as the problem size increases, rendering the problem NP-hard and making it computationally infeasible to find and verify global optimal solutions in a reasonable time.

Two prominent applications are mean-variance portfolio selection problem in Markowitz theory \citep{Markowitz1952} and the quadratic error minimization problem in linear regression, with a sparsity constraint on the number of non-zero decision variables in each case. In the sparse portfolio selection problem, a sparsity constraint on the number of assets to include in the portfolio is an important extension of the standard Markowitz theory in practice since it prevents dense negligible weight allocations across the asset universe by restricting the scope to a subset portfolio. Moreover, an extensive cover of the portfolio can be regarded as passive index-tracking for portfolio managers, which is disfavored by clients when an active management fee is charged \citep{Bertsimas2022}. While the method is developed in the portfolio selection setting, it extends naturally to other convex MIQP subset selection problems like sparse feature selection in linear regression where the sparsity constraint improves model interpretability \citep{tibshirani1996regression} and portability \citep{han2015deep}.
 The only difference between the two applications is the additional aggregate-weight equality constraint and an optional box constraint on individual asset weights in the portfolio selection problem.

The standard Markowitz model with only aggregate portfolio weight constraint is fast to solve as a convex quadratic program (QP) with the weights on individual assets being the decision variable. However, in practice, additional constraints can be imposed to regulate the solution such as box constraints that restrict weights on single assets within specified lower and upper bounds, and a sparsity constraint to prevent excessive trivial allocations and the incurred transaction cost \citep{Boyd2024}. While it is easy to incorporate the box constraints that only makes the problem a constrained QP, the introduction of a sparsity constraint complicates the problem by introducing binary decision variable, $s$, of whether including an asset in the subset portfolio or not, therefore discretizing the feasible set. With the discrete sparsity constraint, $\|x\|_0=k$, being considered, the problem becomes a MIQP, which is NP-hard so that the computation cost increases exponentially as the problem size increases \citep{Floudas2025}. Based on the B\&B (Branch-And-Bound) algorithm \citep{LandDoig1960}, the main-stream commercial solvers like Gurobi and CPLEX can find the global optimal solution in theory but fail to do so in reasonable times for large-scale problems. In such scenarios, practitioners must rely on high-quality incumbent solutions under strict time budgets.

 DASH attacks the combinatorial nature of the problem by performing a systematic dimensionality reduction of the problem. It filters out a subset of the entire asset universe and therefore identifies a low-dimensional problem, on which Gurobi can be applied as usual. Since the attention is restricted to a subset of assets, the global optimal solution in the subproblem can be no better than that of the original problem. Therefore, the trade-off lies at filtering out a subset while imposing minimum impairment on the achievable global optimal solution. DASH achieves such a dimensionality reduction by working on the continuous relaxation (\ref{P1.4}) of a discrete MIQP formulation of the problem (\ref{P1.3}) where the sparsity constraint is handled by a binary decision variable $s\in \{0,1\}^p$, and so by its continuous counterpart $t\in[0,1]^p$ in the relaxed problem. Since the feasible set is a polyhedron, we use the Frank-Wolfe (FW) algorithm \citep{FW1956algorithm} to perform constrained gradient descent on the value function of the relaxed problem parametrized by $t$. We introduce an endogenous pruning threshold such that, as the FW gradient descent progresses, variables whose corresponding relaxed decision variable $t_i$ fall below this threshold are systematically discarded. This effectively guides the algorithm to focus on a subset of assets with high potential of yielding global optimal or comparable portfolio solutions.

We make two contributions in this paper. First, we develop a specific dimensionality reduction scheme, DASH for subset portfolio selection problems, which can also be applied to other subset selection problems that can be formulated as MIQP as \ref{P1.2}. We show the improvement of DASH over Gurobi in terms of both the quality of incumbent solutions and the duration of improvement for difficult problem configurations. Second, for a general subset selection problem with a sparsity requirement of $k$ and a full set $U$, $k<|U|=p$, we establish that restricting B\&B search to an active subset $S\subset U$, $k<|S|=m<p$ reduces the combinatorial feasible set from $\binom{p}{k}$ to $\binom{m}{k}$, providing a theoretical justification for dimensionality reduction as a general strategy for improving incumbent solution quality in large-scale convex MIQPs under a fixed time budget.

The paper is structured as the following: Section \ref{sec:litreview} provides a literature review; Section \ref{sec:theory} formulates the subset portfolio selection problem as a relaxed MIQP and derives its theoretical properties; Section \ref{sec:algorithm} develops the DASH algorithm and shows its applicability to other subset selection problems; Section \ref{sec:numerical} conducts the numerical experiment to compare the performance of DASH and Gurobi; Section \ref{sec:conclusion} concludes the paper.

\section{Literature Review}\label{sec:litreview}
    The Markowitz portfolio selection problem is an established theory guiding investments as it captures the fundamental idea of the trade-off between risk and return, which is characterized as a quadratic minimization problem of the difference between the variance and the mean of asset returns \citep{Markowitz1952}. While the preference of a high return is universal, investors can have different preferences of the risk involved in the investment. This can be formulated into the problem as a risk-aversion parameter that scales either the variance or the return for an effective comparison. That is, for a unit amount of risk, measured by the portfolio variance, investors with different risk preference have different prices, measured by the percentage of return, for them to take on the risk.

    A practical concern in implementing Markowitz theory is the magnitude of the condition number of the sample covariance matrix, defined as the ratio between its largest and smallest eigenvalues, reflecting the spread between the dominant systematic risk factor and the smallest idiosyncratic noise components respectively. Empirical sample covariance matrices are often ill-conditioned, reflecting a substantial dispersion in their eigenvalue spectra. Mathematically, unconstrained optimization using an ill-conditioned sample covariance matrix results in highly unstable portfolios with extreme long and short positions on individual assets. In the literature, there are two contrasting perspectives regarding the root cause of this phenomenon.

    On the one hand, the statistical perspective views this ill-conditioning primarily as a result of estimation error. Since empirical matrices are estimated from finite financial time series with $p$ assets and $T$ observations, the eigenvalue spectrum becomes systematically distorted when $p/T$ is non-negligible, where the largest eigenvalues are inflated while the smallest eigenvalues are deflated toward zero \citep{LalouxEtAl1999}. Upon inversion of the covariance matrix $\Sigma$, the artificially small eigenvalues produce excessively large eigenvectors in $\Sigma^{-1}$, which causes the optimizer to place excessive weight on the low-variance directions that are in fact dominated by estimation noise, leading to an "error maximizing" behaviours \citep{Michaud1989}. Consequently the portfolio solution is also sensitive to small changes in the input data, rendering the solution less robust and trustworthy. On the other hand, a structural perspective argues that the extreme weights are not merely statistical artifact, but rather a rational optimizer choice due to the existence of a dominant market factor, enabling efficient factor and residual risk elimination on a subset of highly correlated assets \citep{green1992when}.
    To mitigate the issues related to an ill-conditioned sample covariance matrix under the first statistical perspective, researchers have developed regularization techniques, such as shrinkage estimators, to obtain well-conditioned surrogate of the raw sample covariance matrix \citep{LedoitWolf2003, BaiShi2011,ledoit2012nonlinear}. In addition, explicit constraints (no short-selling) limiting individual asset exposure have also been shown to mathematically induce a shrinkage-like effect on the covariance matrix, effectively reducing sampling error even if the constraints themselves are empirically incorrect\citep{JagannathanMa2003}. The authors also note that the capability of the regularization methods depends on whether the reduction in sampling error outweighs the specification error introduced. However, from the structural perspective, if the matrix is well-estimated, the specification error introduced by applying biased shrinkage estimators or imposing empirically incorrect constraints prevents the optimizer from sufficiently hedging the dominant market factor.

    Ultimately, whether viewing an ill-conditioned covariance matrix as a statistical artifact that requires regularization or an inescapable structural reality of the market depends on the specific context and the available time series data. In this paper, the magnitude of the condition number matters to us from a distinct computational angle: an ill-conditioned covariance matrix possesses a wider spectrum of eigenvalues, which makes the associated MIP problem more difficult to solve. To conduct a comprehensive analysis of the performance of the proposed DASH method, in comparison with standard Gurobi, on problems with different configurations, we use factor models \citep{Ross1976,BaiShi2011} to generate synthetic covariance matrices with precise control over the level of condition numbers.

    Two constraints central to this paper are box constraints on individual asset weights and a sparsity constraint on the number of assets held in the portfolio \citep{chang2000heuristics}. As established above, box constraints with controlled shorting strike a balance between hedging flexibility and robustness against "error-maximizing" behaviour. The sparsity constraint prevents dense weight allocation across all assets as in standard Markowitz model, which is undesirable in practice since it incurs extra transaction cost and monitoring cost for portfolio managers. Such a broadly-spread portfolio can also be criticized as passive index-tracking when the portfolio is expected be actively managed by fund managers \citep{Bertsimas2022}.

    The constrained sparse portfolio selection problem transforms the standard convex QP Markowitz model into a MIQP problem, which is first formulated by \cite{Bienstock1996}. This reclassification to MIQP makes the problem exceptionally difficult to solve since the sparsity constraint fractures the convexity of the feasible set. To solve complex MIQPs, modern commercial solvers such as Gurobi and CPLEX rely on the B\&B framework \citep{LandDoig1960}, where various cutting plane techniques like Outer Approximation (OA) and Generalized Benders Decomposition (GBD) may be deployed for efficient branching strategies\citep{Floudas2025}. In a MIQP, the algorithm develops a binary search tree of feasible solutions and prune out branches incapable of containing the global optimum by comparing the lower bounds of continuous relaxations of the problem at each node with the upper bounds provided by the incumbent solutions.

    Two inherent challenges of the B\&B algorithm lie in development of effective branching strategies and the exponential growth of the feasible set as the problem size increases. Gurobi handles the former well by deploying sophisticated heuristics to find and improve incumbent solutions quickly so that good-quality solutions can be obtained in the early stages of the run \citep{gurobi}. However, the latter challenge is more innate to the algorithm as it obtains optimality certificate by exhausting the binary tree with pruning, which, as a representation of the discrete feasible set, is computationally intractable for large-scale problems. This curse of dimensionality is an inescapable reality faced by exact B\&B algorithms. When solving a large-scale MIQP problem in practice, it is often infeasible to find and verify global optimum in reasonable time since, as the computation continues, the rate and magnitude of objective improvements decay exponentially, and the solver begins spending a disproportionate amount of time on improving the lower bound to close the gap with the incumbent such that a global optimality certificate is produced. Consequently, in practical environments subject to strict latency constraints, practitioners need to predefine the time budget, settling for high-quality local solutions rather than squandering the computation budget on marginal improvement of the solution and obtaining a certificate of global optimality. Because of the computational intractability of exact B\&B algorithms, early literature worked on bypassing the NP-hardness in MIQP formulation by designing greedy metaheuristic algorithms to navigate the entire $\binom{p}{k}$ search space. \cite{chang2000heuristics} propose Genetic Algorithms, Tabu Search, and Simulated Annealing, which are fundamentally randomized local search algorithms. These metaheuristics provide local solutions computationally fast but the solution quality can be unpredictable due to its insufficient exploration of the entire feasible set.

    DASH acknowledges the NP-hardness of MIQP and
    aims at addressing the second concern by filtering out a subset of the entire asset universe on which MIP solvers can be applied as usual. By forfeiting the infeasible goal of finding and verifying the global optimal solution of a large-scale problem, we can facilitate MIP solvers in finding better incumbent solutions under a limited time budget by restricting the scope of exact B\&B search on a subproblem. Different from metaheuristics based on local search algorithms \cite{chang2000heuristics}, DASH is essentially a dimensionality reduction operator on the original asset space with comprehensive consideration of the entire search space based on the continuous relaxation \ref{P1.4} of a MIQP formulation \ref{P1.3}. The model formulation is similar to \citep{moka2025scalable}. Different from the authors' progressive choice of the auxiliary parameter $\delta$, we only set the auxiliary parameter $\delta$ to be a negligible positive number for a regularization purpose as will be explained in \ref{P1.3}. The philosophy of our approach is similar to the one adopted by \citep{FischettiLodi2003}, where the authors improve the early stage MIP solver behavior with a hierarchical scheme, where a high-level strategic branching is introduced as a metaheuristic on top of the low-level tactical branching, standard black-box MIP solver, conducted by CPLEX. The asset filtering process in our algorithm is a strategic and systematic choice of a subproblem in which assets with better attributes, in terms of their return, volatility, and covariance with the other assets, are kept such that computations performed by MIP solvers are more worthwhile without dwelling on poor-quality assets.

    The asset filtering process is achieved by performing FW gradient descent \citep{FW1956algorithm} on the value function of the relaxed problem, \ref{P1.4}. Originally proposed for QP on polytopes, the Frank-Wolfe algorithm updates the current point by taking its convex combination with a selected vertex via a linear minimization oracle, ensuring iterates remain feasible automatically within $\mathcal{T}$. Since the vertices of the polytope $\mathcal{T}$ are exactly the discrete feasible points $s\in \mathcal{S}$, each update geometrically pulls the current iterate towards a binary asset selection, producing a natural update trajectory that directly informs a pruning scheme. On the other hand, an alternative efficient gradient descent method is the projected gradient descent on $\ell_1$-balls \citep{Duchi2008}, but it can produce infeasible updates outside the feasible set $\mathcal{T}$ especially when large gradient components are present and exact line search is infeasible.

    There are two practical concerns on the FW algorithm. First, because the value function in \ref{P1.4} is generally not convex over $\mathcal{T}$, the convergence properties established for FW under convexity \citep{Jaggi2013} do not directly apply to our setting. Therefore, the theoretical justification for FW is based on its structural fitness rather than its asymptotic optimality properties. Second, the oscillatory "zig-zagging" behaviour near the boundary of $\mathcal{T}$ is a well-known issue of the standard FW algorithm that undermines its convergence performance \citep{LacosteJulien2015}. Though \cite{LacosteJulien2015} propose the away-step FW to address the second concern, we circumvent both limitations of the standard FW by design. Since the goal is not to find the continuous minimizer $t^*$ precisely, but rather to identify a reliable partition of the asset universe into discarded and reserved assets, we grant only a small but sufficient number of FW iterations well before convergence precision issues. We establish and validate the filtering process through FW gradient descent with an early-stopping criterion when we develop the algorithm.

    Although sparse subset selection problems formulated as MIQP are NP-hard, a standard alternative approach in statistical learning is to include regularization terms in the objective function such as Lasso ($\|x\|_1$), Ridge ($\|x\|_2$) or a combination of the two as an Elastic Net to force sparse solutions without explicitly imposing a sparsity constraint. The same regularization techniques are also adapted to the sparse portfolio selection problem to regulate the sample covariance matrix, tightening continuous relaxations of the problem so that more efficient branch-pruning can be achieved. The authors focus on a surrogate of the original problem, which replaces the raw sample covariance matrix $\Sigma$ with $\Sigma +\frac{1}{\gamma}I$ \citep{Bertsimas2022}, effectively a uniform diagonal perturbation to the sample covariance matrix. This replacement performs a rigid shift of the entire eigenvalue spectrum of $\Sigma$ to the right, bounding the smallest eigenvalue away from zero, reducing the condition number of $\Sigma$. As the magnitude of the regularization term increases, the diagonal perturbation forces the matrix structure towards strict diagonal dominance, yielding stable computational behaviors. However, by inflating the variances, the augmentation is artificially depressing the relative correlations between assets, leading the optimizer to treat the  assets as more independent than what the sample data indicates, potentially underestimating the true risk to which a portfolio is exposed. Ultimately, this presents a classic bias-variance trade-off in estimating the population covariance matrix. The sample covariance matrix is unbiased but ill-conditioned in large-scale problems, while structured estimators sacrifice unbiasedness in seek of improved conditioning \citep{LedoitWolf2004}.

\section{Theory}\label{sec:theory}
We first introduce the transformation from simple QP (\ref{P1.1}) to MIQP formulation (\ref{P1.2}) when a sparsity constraint is imposed. Then we develop a regulated equivalent formulation \ref{P1.3} and its continuous relaxation \ref{P1.4}. Then we derive relevant theoretical properties of \ref{P1.4}, and use them to motivate and support the development of DASH. Lastly, we provide an abstract characterization of a class of subset selection problems, providing a unified high-level perspective and showing the general applicability of DASH to other applications that can be identically formulated as \ref{P1.4}.

\subsection{A general framework}
Consider the following general constrained convex QP problem:
\begin{align*}\label{P1.1}
    \underset{x\in \mathcal{B}}{\min} &\;\; f(x)=x^\top \Sigma x - \mu^\top x \tag{P1.1}\\
    s.t.&\;\; Gx -b \preceq 0\\
    &\; \; Ax - c = 0
\end{align*}
where $\Sigma \in \mathbb{R}^{p\times p}$ is positive semi-definite (it can be further assumed to be positive definite), $\mu\in \mathbb{R}^p$, $G\in \mathbb{R}^{m\times p}$, $b\in \mathbb{R}^{m}$, $A\in \mathbb{R}^{n\times p}$, and $c\in \mathbb{R}^{n}$. We note that the domain of $x$ can be restricted onto an arbitrarily large compact set $\mathcal{B}\subset \mathbb{R}^p$ instead of $\mathbb{R}^p$ since it is infeasible to have infinite short or long position in a single asset given finite level of economic activities and borrowing capacity. Furthermore, when a solver yields a solution with excessive weight allocations, it places unconditional trust on the estimated data, leading to solutions that are sensitive to estimation error and prone to 'error-maximizing' behaviors \citep{Michaud1989}. Therefore, without loss of generality, we let $\mathcal{B}:=[-M,M]^p$ where $M$ can be chosen arbitrarily large so that only its compactness matters.
In addition, denote the feasible set as
\[X=\{x\in \mathbb{R}^p: Gx-b\preceq 0, \; Ax-c=0\}.\]
We require the inequality constraints to form box constraints that impose lower and upper bounds on the decision variable $x$. A typical set of constraints that we consider in the mean-variance portfolio setting are the box inequality constraints on the portfolio weight, $l\preceq  x\preceq u$ ($l\prec u$), and an aggregate weight equality constraint $\mathbf{1}^\top x=1$. They can be standardized as \ref{P1.1} by defining
\[G = \begin{bmatrix}
    -I\\
    I
\end{bmatrix}\in \mathbb{R}^{2p\times p}; \;\; b =\begin{bmatrix}
    -l\\
    u
\end{bmatrix}\in \mathbb{R}^{2p},\]
and
\[A = \mathbf{1}^\top \in \mathbb{R}^{1\times p};\;\; c =1\in \mathbb{R}^{1}.\]

By adding a sparsity constraint on $x$, the problem becomes a subset selection problem, which is a MIQP (Mixed-Integer Quadratic Program):
\begin{align*}\label{P1.2}
    \underset{x\in \mathcal{B}}{\min} &\;\; f(x)=x^\top\Sigma x - \mu^\top x \tag{P1.2}\\
    s.t.&\;\; Gx -b \preceq 0\\
    &\; \; Ax - c = 0\\
    &\;\; \|x\|_0 = k.
\end{align*}
To have a clear systematic analysis of \ref{P1.2}, we disentangle the sparsity attribute from the continuous decision variable $x$ and characterize it separately by introducing a binary decision variable $s\in \{0,1\}^p$, leading to the following parametrized optimization:
\begin{align*}\label{P1.3}
    \underset{s\in \{0,1\}^p}{\min}\underset{x\in \mathcal{B}}{\min} &\;\; f(s,x)=f(T_s x)=x^\top T_s \Sigma T_sx - \mu^\top T_s x \tag{P1.3}\\
    s.t.&\;\; G T_s x- b \preceq 0\\
    &\;\; AT_s x -c = 0\\
    &\;\; \|s\|_0 = \|s\|_1 = k
\end{align*}
where $T_s = diag(s)$. Denote the feasible set of $s$ as $\mathcal{S} = \{s\in \{0,1\}^p : \|s\|_1 = k\}$. For a fixed $s$, the operator $T_s$ is an orthogonal projection from $\mathbb{R}^p$ onto the subspace $\mathbb{R}^{k}(s)$ where we note the dependency of the subspace on $s$. As such, the subproblem in \ref{P1.3} is a high-dimensional embedding of \ref{P1.1} and we can write its low-dimensional equivalent as:
\begin{align*}\label{P1.3.1}
    \underset{x\in \mathcal{B}^{|s|}}{\min} &\;\; f(x_{[s]})=x_{[s]}^\top \Sigma_{[s]}x_{[s]} - \mu_{[s]}^\top x_{[s]} \tag{P1.3.1}\\
    s.t.&\;\; G_{[s]}x_{[s]} - b  \preceq 0\\
    &\;\; A_{[s]}x_{[s]}-c=0
\end{align*}
where $x_{[s]}, \mu_{[s]}\in \mathbb{R}^{k}$ are sub-vectors of $x$ and $\mu$, preserving active coordinates of $s$, i.e. $\text{supp}(s) =\{i\in \{1,\dots, p\}: s_i \neq 0\}$; $\Sigma_{[s]}\in \mathbb{R}^{k\times k}$ is the sub-matrix of $\Sigma$ without $i$th row and column for $i\in \text{supp}(s)^c$; $G_{[s]}\in \mathbb{R}^{m\times k}$, $A_{[s]} \in \mathbb{R}^{n\times k}$ are sub-matrices of $G$ and $A$ without $i$th column for $i\in \text{supp}(s)^c$, and $\mathcal{B}_{[s]}:=[-M,M]^k\in \mathbb{R}^k$ is the restriction of $\mathcal{B}$ to the active coordinates specified by $s$.

Given a fixed $s$, the subproblem \ref{P1.3.1} is the same as the basic problem \ref{P1.1}. Suppose $(s^*,x^*)$ is a global minimizer of \ref{P1.3}, it is safe to assume that $s^*$ is unique. However, $x^*$ is not unique in general due to the lack of curvature in the non-trivial null space of the projection $T_s$. As such, we modify $\Sigma$ by adding a positive semidefinite diagonal matrix $\delta (I-T_s^2)$, $\delta>0$, to make the objective function strictly convex in the subspace $\ker(T_s)$ for all $s\in \mathcal{S}$, recovering uniqueness of solution $(s^*,x^*)$ to \ref{P1.3}.

Denote the modified matrix as $\tilde{\Sigma}_s= T_s\Sigma T_s +\delta (I-T_s^2)$, and consider the continuous relaxation of discrete feasible set $\mathcal{S}$ to $\mathcal{T} = \{t\in [0,1]^p : \|t\|_1 = k\}$. The continuous relaxation of \ref{P1.3} can be rewritten as
\begin{align*}\label{P1.4}
    \underset{t\in [0,1]^p}{\min}\underset{x\in \mathcal{B}}{\min} &\;\; f(t,x)=x^\top \tilde{\Sigma}_tx - \mu^\top T_t x \tag{P1.4}\\
    s.t. & x\in X(t):=\{x\in \mathbb{R}^p: GT_tx-b\preceq 0, AT_tx-c=0\}\\
    &\;\;  \|t\|_1=\sum_{i=1}^p t_i = k,
\end{align*}
whose subproblem, given a fixed $t$, is
\begin{align*}\label{P1.4.1}
    \underset{x\in \mathcal{B}}{\min} &\;\; f(x;t)= x^\top \tilde{\Sigma}_tx - \mu^\top T_t x \tag{P1.4.1}\\
    s.t.&\;\; x\in X(t).
\end{align*}

Continuous relaxation in \ref{P1.4} convexifies $\mathcal{S}$ as $\mathcal{T}$, so that we can make uses of continuous optimization techniques. It is also worth-mentioning that the correspondence $X:\mathcal{T}\rightarrow \mathcal{B}$ shows the dependency of the feasible set $X(t)$ in \ref{P1.4.1} on the choice of $t\in \mathcal{T}$.

\subsection{Theoretical Results}
\ref{P1.4} is the backbone formulation that enables the proposed DASH method. We start by deriving necessary theoretical results that facilitate and support the method. Firstly, we specify Assumption \ref{assp1} under which the feasible set $X(t)$ is regularized. After establishing properties of the subproblem \ref{P1.4.1} as a QP for a fixed $t\in \mathcal{T}$ in Lemma \ref{lem1}, we use the Maximum Theorem\citep[pp.237]{sundaram1996first}, adapted as the Minimum Theorem for minimization problems in Theorem \ref{theo1}, to prove continuity of the solution set $(f^*,x^*,\lambda^*,\nu^*)$ in \ref{P1.4.1} as functions of $t\in \mathcal{T}$, extending the analysis to the outer problem \ref{P1.4}. Theorem \ref{theo3} derives the gradient expression of the value function $f^*$ in \ref{P1.4}, which is equivalent to an application of the Envelop theorem. Theorem \ref{theo4} further analyzes the sign and magnitude of gradient elements of $f^*$, providing a concrete justification for the development of DASH method.

\begin{assump}\label{assp1}
    Given a $t\in\mathcal{T}:=\{t\in[0,1]^p: \|t\|_1=k\}$, denote the index set of coordinates with active inequality constraints by $\mathcal{I}(x):= \{i\in \{1,\dots, p\}: l_i-t_ix_i=0\;or\;t_ix_i-u_i=0\}$
    for $x\in X(t)$. We assume that $X(t)$ is non-degenerate. That is, for all $x\in X(t)$, $\text{supp}(t)\backslash \mathcal{I}(x)\neq \emptyset$.
\end{assump}

Firstly, for any $t\in\mathcal{T}$, if $t_i=0$, then $l_i<t_ix_i<u_i$, so that $\mathcal{I}(x)\subseteq \text{supp}(t)$ and for brevity, we keep the notation $\mathcal{I}(x)$ but note its implicit dependence on $t$. For a given $t\in \mathcal{T}$, the assumption asserts that there exists $j\in \text{supp}(t)$ such that $l_j<t_jx_j<u_j$ so that both the lower and upper bounds on $x_j$ are inactive. This holds for all but a degenerate case, where an exact solution of the overdetermined linear system of active inequality constraints in $\mathcal{I}(x)$ and the aggregate-weight constraint exists. The existence relies on an artificial algebraic relation between $l$, $u$, and $t$ that is not satisfied in general. Since the lower and upper bounds on $x_i$ cannot be active simultaneously, $\mathcal{I}(x)$ is a characterization of all active inequality constraints at $x\in X(t)$. Suppose $X(t)$ is degenerate at $x$ i.e. $\mathcal{I}(x)=\text{supp}(t)$, then given $|\mathcal{I}(x)|$ equalities $l_i-t_ix_i=0$ or $t_ix_i-u_i=0$, $x_i=\frac{l_i}{t_i}$ or $x_i = \frac{u_i}{t_i}$ for $i\in \mathcal{I}(x)$ are uniquely determined by $l$ and $u$. Now consider the aggregate weight equality $A^\top T_tx-c=\sum_{i\in \text{supp}(t)}t_ix_i-1=\sum_{i\in\mathcal{I}(x)}t_ix_i-1=0$. This is only satisfied if $l$ and $u$ are chosen so that the active lower and upper bounds add up to 1. Regarding the active lower and upper bounds as a vector $b\in \mathbb{R}^{|\mathcal{I}(x)|}$, given the Lebesgue measure $\psi$ on the measure space $(\mathbb{R}^{|\mathcal{I}(x)|},\mathcal{B}(\mathbb{R}^{|\mathcal{I}(x)|}))$, we have that $\psi(\{b\in\mathbb{R}^{|\mathcal{I}(x)|}: \mathbf{1}^\top b=1\})=0$ i.e. the event of degenerate configurations has Lebesgue measure zero. As such, an arbitrarily small perturbation to the bounds will lead to immediate violation of the relation, restoring non-degeneracy. Assumption \ref{assp1} therefore excludes this measure-zero set of degenerate configurations and is satisfied in general.

\begin{lemma}\label{lem1}
    In \ref{P1.4}, for any $t\in \mathcal{T}$, $\delta\in \mathbb{R}_{++}$, we have the following results on the subproblem \ref{P1.4.1}
    \begin{enumerate}
        \item The subproblem $\ref{P1.4.1}$ is a convex optimization problem. In addition, if $\Sigma$ is positive definite or $t\prec \mathbf{1}$, component-wise, then it is a strictly convex optimization problem and the solution $x^*$ is unique for a given $t\in \mathcal{T}$, and $x^*_i = 0$ if $t_i = 0$.
        \item LICQ holds for all $x\in X(t)$.
        \item The Lagrangian of \ref{P1.4.1} is
        \[L(x,\lambda, \nu;t)=x^\top \tilde{\Sigma}_t x - \mu^\top T_t x + \lambda^\top (GT_tx-b)+ \nu ^\top (AT_t x-c)\]
        The solution of \ref{P1.4.1}, $x^*(t)$, satisfies the following KKT conditions:
        \begin{enumerate}
            \item Stationarity \[D_x L(x^*,\lambda^*,\nu^*;t)=2(x^*)^\top \tilde{\Sigma}_t - \mu^\top T_t  + (\lambda^*)^\top G T_t +(\nu^*)^\top AT_t=\mathbf{0}^\top\]
            \item Primal feasibility
            \[GT_t x^*-b\preceq 0\]
            \[AT_t x^*-c=0\]
            \item Dual feasibility
            \[ \lambda^*\succeq 0\]
            \item Complementary slackness
            \[(\lambda^*)^\top(GT_t x^*-b)=0\]
        \end{enumerate}
    \end{enumerate}
\end{lemma}

\begin{proof}
    See Appendix \ref{apex1:proof_lem1}
\end{proof}

Define the correspondence $X^*(t): \mathcal{T} \rightarrow P(\mathcal{B})$ by $X^*(t):=\underset{x\in X(t) }{\text{argmin}}f(x;t)$, which maps $t\in \mathcal{T}$ to the solution set of \ref{P1.4.1}. By Lemma \ref{lem1}, the solution  $x^*(t)$ is a singleton for every $t$ so that the correspondence is a function.

In Theorem \ref{theo1}, we rewrite the Maximum Theorem for maximization problems \citep{sundaram1996first} as the Minimum Theorem for minimization problems.
\begin{theorem}\label{theo1}
    Suppose $f$ is a continuous on $\mathcal{T} \times \mathcal{B}$ and $X: \mathcal{T} \rightarrow P(\mathcal{B})$ is a compact-valued continuous correspondence. Let $f^*:\mathcal{T} \rightarrow \mathbb{R}$ and $X^*:\mathcal{T} \rightarrow P(\mathcal{B})$ be defined by
    \[f^*(t)=\min\{f(t,x): x\in X(t)\}\]
    \[X^*(t) = \text{argmin}\{f(t,x): x\in X(t)\}=\{x\in X(t): f(t,x)=f^*(t)\}.\]
    Then $f^*$ is a continuous function on $\mathcal{T}$, and $X^*$ is a usc correspondence on $\mathcal{T}$.\\
    In addition, if $f(\theta,\cdot)$ is convex in $x$ for each $t$, and $X$ is convex-valued, then $X^*$ is a convex-valued correspondence. If $f$ is strictly convex, then $X^*$ is a single-valued upper-hemicontinuous correspondence, hence a continuous function.
\end{theorem}

\begin{theorem}\label{theo2}
    In \ref{P1.4}, if $\tilde{\Sigma}_t$ is positive definite, i.e. $\Sigma$ is positive definite or $t\prec \mathbf{1}$, we have the following results:
    \begin{enumerate}
        \item The minimizer $x^*(t)$ and the minimum $f^*(t)$ of \ref{P1.4.1} are continuous on $\mathcal{T}$.

        \item The optimal dual variables $\lambda^*(t)$, $\nu^*(t)$ are also continuous on $\mathcal{T}$.
    \end{enumerate}
\end{theorem}
\begin{proof}
    See Appendix \ref{apex1:proof_theo2}
\end{proof}

By Lemma \ref{lem1} and Theorem \ref{theo2}, for each $t\in\mathcal{T}$ with the global solutions $x^*(t)$, $\lambda^*(t)$, and $\nu^*(t)$, we have that
\[\mathcal{L}^*(t):=L(x^*(t),\lambda^*(t),\nu^*(t);t)=f^*(t)=\underset{x\in X(t)}{\text{min}} f(x;t)=f(t,x^*(t))\]
where $\mathcal{L}^*$ is the Lagrangian $L$ evaluated at the optimum $(x^*,\lambda^*,\nu^*)$, i.e. the value function $f^*$, for all $t\in \mathcal{T}$ and $x^*(t):=\underset{x\in X(t)}{\text{argmin}}\;f(x;t)$.

We now derive the gradient expression of the value function in Theorem \ref{theo3}, which is consistent with a direct application of the Envelope Theorem. We denote $D_tf$ derivative of $f$ w.r.t. $t$ as in \cite{sundaram1996first}.

\begin{theorem}\label{theo3}
    In \ref{P1.4}, the gradient of the value function w.r.t. $t$, $\nabla_t f^*$, is
    \begin{align*}
        \nabla_t f^* = & \Big[D_tf + (\lambda^*)^\top D_t g^* +(\nu^*)^\top  D_t h^* \Big]^\top\\
    = & [2t^\top T_x(\Sigma -\delta I)T_x -\mu^\top T_x +(\lambda^*)^\top GT_x +(\nu^*)^\top AT_x]^\top
    \end{align*}
\end{theorem}
\begin{proof}
    See Appendix \ref{apex1:proof_theo3}
\end{proof}
Given that $G = \begin{bmatrix}
        -I\\
        I
    \end{bmatrix}$, and $A = \mathbf{1}^\top $, the element-wise gradient expression of the value function is given by
    \begin{align*}
        \nabla_tf^* = \begin{bmatrix}
            2x_1^{*^2}(\sigma_{11}-\delta)t_1 +x_1^* [(\lambda_{p+1}^*-\lambda_1^*)+\nu^* -\mu_1 +2\sum_{i\neq 1}t_i\sigma_{i1}x_i^*]\\
            \vdots \\
            2x_p^{*^2}(\sigma_{pp}-\delta)t_p +x_p^* [(\lambda_{2p}^*-\lambda_p^*)+\nu^* -\mu_p +2\sum_{i\neq p}t_i\sigma_{ip} x_i^*]\\
        \end{bmatrix},
    \end{align*}
so the gradient component of asset $i$ is
    \begin{align*}
        \Big[\nabla_tf^*\Big]_i = 2x_i^{*^2}(\sigma_{ii}-\delta)t_i + x_i^*\Big[(\lambda_{p+i}^*-\lambda_{i}^*)+\nu^* -\mu_i +2 \sum_{j\neq i}t_j\sigma_{ji}x_j^*\Big].
    \end{align*}
We establish the properties of the gradient components in Theorem \ref{theo4}.
\begin{theorem}\label{theo4}
    In \ref{P1.4}, for $t\in \mathcal{T}$, the gradient components of $\nabla_tf^*$ are all non-positive. In addition, $\underset{t_i \rightarrow 0}{\lim} [\nabla_t f^*]_i=0$
\end{theorem}
\begin{proof}
    See Appendix \ref{apex1:proof_theo4}
\end{proof}

As such, we have a full set of results from the uniqueness of the solution bundle $(f^*,x^*,\lambda^*,\mu^*)$ for a given $t\in \mathcal{T}$ in the subproblem \ref{P1.4.1} (Lemma \ref{lem1}) to their continuity as functions of $t\in \mathcal{T}$. Equipped with the gradient expression of the value function $f^*$, constrained gradient descent methods can be applied to solve the outer problem \ref{P1.4} iteratively. The primary claim of DASH is that the minimizer $t^*:= \underset{t\in \mathcal{T}}{\min} f^*(t)$ is indicative of good subset portfolios in terms of its distance to the feasible vertices $s\in \mathcal{S}$.

\section{Algorithm}\label{sec:algorithm}
We develop the DASH algorithm and emphasize its applicability to an entire class of subset selection problems in this section. First, we provide a financial interpretation of the terms in the gradient elements of the value function $f^*$ to motivate the gradient descent step. Then we motivate and develop the asset filtering scheme based on the geometry of $f^*$ over the polytope $\mathcal{T}$, where the combinatorial feasible points $s\in \mathcal{S}$ are vertices of $\mathcal{T}$. Theorem \ref{theo4} provides further theoretical justification for the validity of the filtering scheme in DASH and the pruning threshold mechanism.

\subsection{Desirability of Negative Gradient Elements}
We extend the rationale of Markowitz's mean-variance portfolio theory to the relaxed subset portfolio selection problem \ref{P1.4} to connect the role of negative gradient elements in decreasing the value function and the desirability of including the corresponding assets in a subset portfolio. This provides motivation and support for the gradient descent scheme in the proposed DASH method.
    In \ref{P1.4}, the negativity of $[\nabla f^*]_i$ is indicative of the desirability of asset $i$ in the portfolio. In particular, $\underset{t_i\rightarrow 0}{\lim} [\nabla_t f^*]_i = 0$, i.e. asset $i$ becomes trivial in \ref{P1.4}.
    It is clear that when solving the outer problem of \ref{P1.4}, the gradient is pointing towards the direction of steepest ascent. Therefore, to minimize the mean-variance objective function, the value function, the trajectory formed by the opposite direction of the gradient is most efficient. That is, for a large negative component of $\nabla_t f^*$, a marginal increase of the corresponding $t_i$ reduces the value function $f^*$ effectively, suggesting desirability of the corresponding asset $i$. Different from a binary vector $s\in \mathcal{S}$, which is a definitive choice of a subset portfolio, the relaxed counterpart $t\in \mathcal{T}=\text{conv}(\mathcal{S})$ distributes the attention, the sparsity level $k$ specified, to all assets in the universe. If asset $i$ has good financial attributes, it tends to receive a higher attention score $t_i$, representing a fractional choice decision. Therefore, we narrow the analysis on individual components and calibrate the idea with rationales in the subproblem \ref{P1.4.1}, the basic Markowitz portfolio selection problem.
\begin{table}[h!]
\centering
\caption{Decomposition of $\frac{\partial f^*(t)}{\partial t_i}\big|_{t_i = \theta}$}
\label{table4.1}
\renewcommand{\arraystretch}{1.5}
\begin{tabular}{p{4cm} c c m{4cm}}
\Xhline{1.5pt}
\textbf{Component} & \textbf{Term} & \textbf{Sign} & \textbf{Interpretation} \\
\hline
Idiosyncratic risk
& $2x_i^{*2}\left(\sigma_{ii} - \delta\right)\theta$
& $+$
& Introducing volatility to the portfolio \\

Box sensitivity
& $x_i^*\left(\lambda_{p+i}^* - \lambda_i^*\right)$
& $0 \text{ or } +$
& Shadow gain  \\

\text{Agg. weight sensitivity}
& $x_i^* \nu^*$
& $\pm$
& Shadow price for financing additional portfolio weight ($+$); for hedging with additional weight ($-$) \\

Return
& $-x_i^* \mu_i$
& $\pm$
& Cost of borrowing $(+)$; Gain from investment $(-)$\\

Diversification
& $2x_i^* \sum_{j\neq i} \sigma_{ij} t_j x_j^*$
& $\pm$
& Risk amplification $(+)$; Hedging $(-)$\\
\Xhline{1.5pt}
\end{tabular}
\end{table}

For an individual asset $i$, the key attributes are its mean return, variance, and its covariance with other assets. In general, an asset with a high mean return is always favorable for profitability, but we also want it to have a low variance, the downside risk in particular. On the other hand, if an asset has a low mean return and a low variance, then it is a good borrowing source, enabling additional long positions on other assets. Last but not the least, the covariance structure of the asset universe is the key reason why holding a portfolio is considered initially in achieving a balance between risk and return via diversification. There is no unilateral preference of the covariance structure of an asset with the others. An asset with low covariance with other assets is considered good when holding the same position as other assets does not amplify the exposure to common underlying market factors. In such a case, whether to include the asset or not in the portfolio lies more on its own properties in terms of the return and variance. However, such assets often do not provide good hedging of the market factor risk as the assets with high covariances. Depending on the sign of the covariance, holding the same or opposite positions on assets with high covariance can reduce the volatility of portfolio effectively.

Table \ref{table4.1} summarizes the effect of each term in the gradient. The three attributes of an asset jointly determine the magnitude of the corresponding gradient component. The first term is the idiosyncratic volatility of asset $i$ and is non-negative for $0<\delta\ll\min \{\sigma_{11},\dots,\sigma_{pp}\} $.  $\sigma_{ii}$ dominates the magnitude of the term, so that if asset $i$ has a high variance, it increases the corresponding gradient element, making the asset less favorable.

The second term comes from the box constraint on the portfolio weight. By KKT condition, it is clear that the term is zero when the corresponding box constraints are slacking, and is strictly positive when either the lower bound or the upper bound is binding. The dual variables $\lambda_{i}^*$ and $\lambda_{p+i}^*$ can be interpreted as the strength of the penalty required for the box constraints to be satisfied. Its magnitude is the shadow gain of marginally relaxing the box constraints. It is crucial to note that, a binding box constraint on an asset already suggests its desirability in that the assigned weight should have been larger had the box constraints were not imposed. This result supports an efficiency improvement of DASH by using the theoretical gradient with no box constraints instead of the numerical gradient derived in Theorem \ref{theo3}, where the only difference is the absence of the third box constraint term in the former.

Similar to the second term, the third term is also a product of the primal and dual variables. Depending on the trade-off between volatility and return specified in the objective function, which is captured by the effective returns, $\nu^*$ can be either positive or negative, reflecting how the objective value reacts to a marginal relaxation of the equality constraint. Given the fact that $\frac{\partial f^*}{\partial \epsilon}=-\nu^*$ from the perturbed problem with equality constraint $AT_t x-c=0$, a marginal relaxation of $\epsilon$ requires assigning this additional weight to any assets. If the objective function is risk-averse, low effective return, then $\nu^*<0$ and $\frac{\partial f^*}{\partial \epsilon} >0$, assigning $\epsilon$ properly increases the objective value. Conversely, if the objective function is risk-seeking, high effective return, then $\nu^* >0$ and $\frac{\partial f^*}{\partial \epsilon} <0$, assigning $\epsilon$ properly reduces the objective value. As such, if the objective function is risk-seeking, i.e. $\nu^*>0$, then $x^*_i \nu^*$ is positive(negative) if $x^*_i$ is positive(negative). That is, the assets that are being shorted ($x^*_i<0$) are attractive ($x^*_i\nu^*<0$) because they provides the additional funds for extra long positions, while the assets that are being longed are made less valuable ($x^*_i\nu^*>0$) since it encapsulates the increased exposure to idiosyncratic risk of a single asset and the opportunity cost of forfeiting other profitable investment opportunities. Since $\nu^*$ is uniform, the magnitude of the third term solely depends on the weight assigned to the assets.

The fourth term considers the individual return aspect of the assets. When longing(shorting) an asset, the term is negative (positive). It evaluates the gain from longing an asset or the cost from shorting an asset.

The last term complicates the problem by considering a weighted covariance structure of the entire asset pool. It aggregates the covariance of the asset $i$ with other assets through the modified covariance $t_j\sigma_{ij}$. When the covariance between two assets are positive(negative), holding the opposite (same) position provides hedging of systematic risk, effectively reducing the variance of the portfolio.

\begin{figure}[h]
    \centering
    \includegraphics[width=1\linewidth]{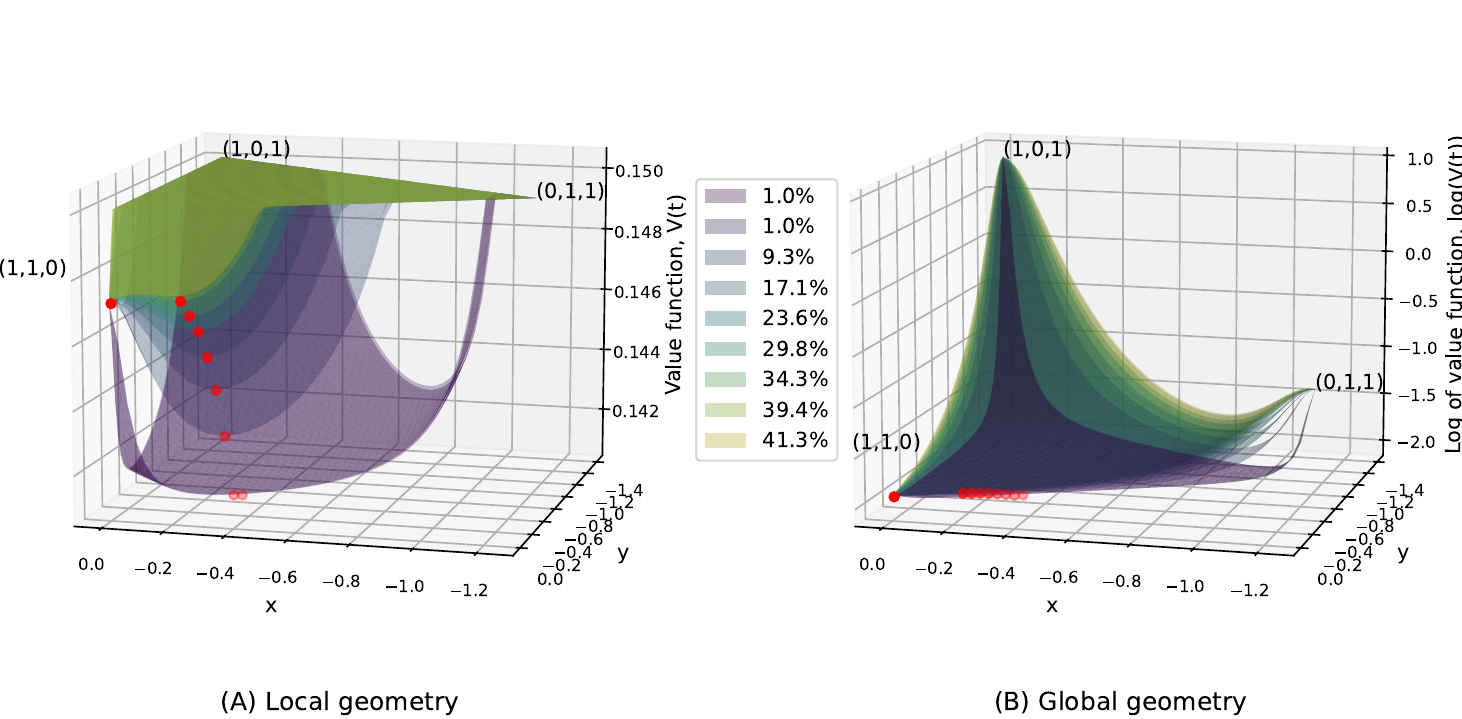}
    \caption{Evolution of Value Function And Minimizer by $\delta$}
    \label{fig4.1}
\end{figure}

The insight above provides a practical account of performing gradient descent on the value function $f^*:\mathcal{T}\rightarrow \mathbb{R}$ to find a minimizer $t^*=\underset{t\in\mathcal{T}}{\text{argmin}} f^*(t)$. Furthermore, a geometric and physical insight of the gradient descent step over $t$ is illustrated in Figure \ref{fig4.1}. The figure shows the landscape of the value function $f^*(t)=\underset{x\in X(t)}{\min} f(t,x)$ in \ref{P1.4} over $\mathcal{T}:=\{t\in [0,1]^3: \|t\|_1=2\}$ in a sample 3-choose-2 problem, where the evolution of the function landscape is corresponding to an increasing sequence of $\delta$ in $\tilde{\Sigma}_t$ as percentages of $\min\{\lambda_{\min}, \min \{diag(\Sigma)\}\}$ and the minimizer $t^*(\delta)$ is marked by a red point for each $\delta$. It is clear that at the feasible vertices, $s\in \mathcal{S}$, the function values agree with the discrete predecessor \ref{P1.3} and the original formulation \ref{P1.2}, which is invariant to the choice of $\delta$. Therefore, the feasible binary vertices serve as anchors of the function landscape, and their corresponding values are what is of interest for comparison in the original subset portfolio selection problem \ref{P1.2}. Suppose complete control is given on the feasible vertices. By raising the function value corresponds to one vertex, it is physically intuitive to infer that the 'gravitational pull' from gradient descent will push the minimizer $t^*$ away from the raised vertex. This phenomenon is illustrated in panel (B) Figure \ref{fig4.1} where the value for $(1,0,1)$ is the highest, i.e. the 2-asset subset portfolio forfeiting asset $2$ is least favored, and the minimizer $t^*$ is pushed away to the vicinity of the opposite boundary corresponding to $\{t_2=1\}\cap \mathcal{T}$, suggesting the potential of achieving better objective values by including asset $2$. Though less drastic, given the choice of asset $2$, the same phenomenon is also clear in the subproblem that determines whether to choose asset $1$ or asset $3$ as the second asset on the subset $\{t_2=1\}\cap\mathcal{T}$. Since $(0,1,1)$ is equipped with a higher value than $(1,1,0)$, the minimizer $t^*$ is closer to $(1,1,0)$ in norms and eventually slides to $(1,1,0)$ as $\delta$ increases.

\subsection{Frank-Wolfe Gradient Descent Scheme}
\begin{figure}[h]
    \centering
    \includegraphics[width=1\linewidth]{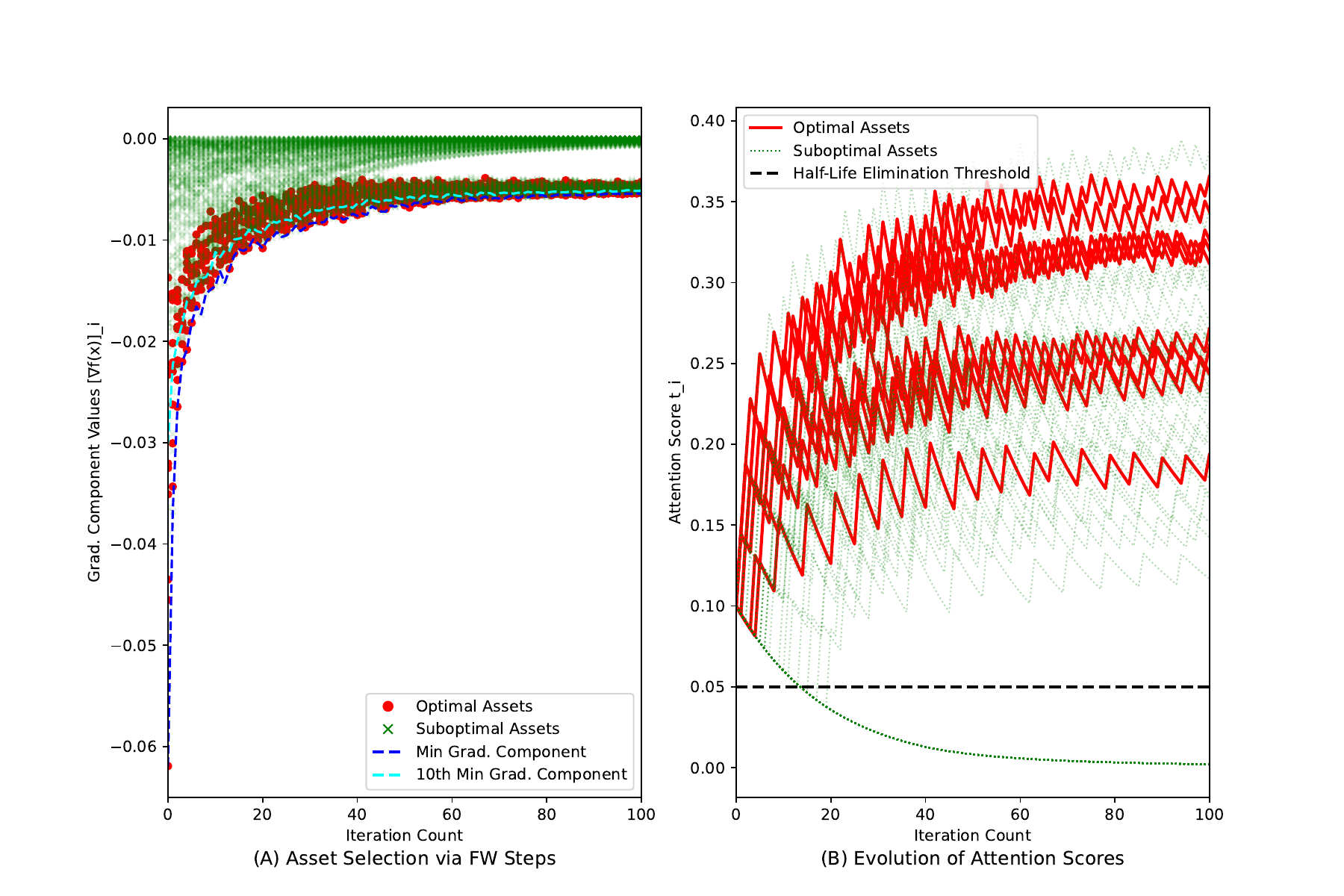}
    \caption{Attention Evolution via FW Steps}
    \vspace{2pt}
    {\raggedright \footnotesize \textbf{Note:} The Figure uses synthetic spiked covariance matrix with the following problem parameters: $p=100$, $k=10$, $l=-1$, $u=1$, $\kappa =1e1$, sample $1$ in seed $42$.\par}
    \label{fig:fw_evolution}
\end{figure}

We have shown both the financial rationale and the geometrically and physically inspired idea of the potential of identifying the global optimal or comparable subset portfolios by performing gradient descent on the value function in \ref{P1.4}. We now develop an efficient gradient descent scheme based on the Frank-Wolfe algorithm.

Since the feasible set $\mathcal{T}$ is a polytope formed by the intersection of the unit hypercube and the sparsity constraint, $[0,1]^p\cap \{\|t\|_1=k\}$, a natural choice of constrained gradient descent methods is the Frank-Wolfe (FW) algorithm. An alternative algorithm is the projected gradient method on to $\ell_1$-balls \citep{Duchi2008}, which is also efficient but requires further treatment to prevent violation of the unit hypercube constraint $t\in[0,1]^p$.

\begin{algorithm}[h]
    \caption{\text{FWF}($\Sigma,\mu, l,u,k,\delta, \alpha,\tau$)}\label{alg1}
    \begin{algorithmic}
    \Require{$\Sigma\in \mathbb{R}^{p\times p}$; $\mu,l,u\in\mathbb{R}^p$, $l\prec0\prec u$; $k\in\mathbb{N}$; $\alpha,\tau \in\mathbb{R}$}
    \Ensure{$t^*\in\mathbb{R}^{p\times p}$}

        \State $p\gets \Sigma.\text{shape}[0]$
        \State $t\gets k/p \cdot \mathbf{1}$\Comment{Initiate $t$ as the midpoint of $\mathcal{T}$}
        \State $P\gets (\Sigma,\mu, l,u,k,t,\delta)$\Comment{Initialize object $P$}
    \For{$n=1,\cdots, \text{ceil}(\log(0.5)/\alpha)$} \Comment{\text{ceil}($x$) takes the smallest integer larger than $x$}
        \State $x^*,\lambda^*, \nu^*, f^* \gets \text{solve P }$ \Comment{simple CQP (\ref{P1.4.1})}

        \State Compute $\nabla_t f^*$ \Comment{By Theorem \ref{theo3}, theoretical or numerical}
        \State $ s \gets \text{Boolean vector with 1's on the smallest k elements of}\;\nabla_t f^*$
        \Comment{FW step}
        \State $t \gets (1-\alpha) t +\alpha s$
        \State $\mathcal{K}\gets \{i\in \{1,\dots,p\}: t_i\leq \tau\}$
        \If{ $\mathcal{K}\neq \emptyset$}
        \State break
        \EndIf
        \State $P.t \gets t$
        \State update $P$ \Comment{update all quantities depending on $t$}
    \EndFor
    \State \Return{$t$}
    \end{algorithmic}
\end{algorithm}

In the convex feasible set $\mathcal{T}=\text{conv}(\mathcal{S})$, the FW method selects a vertex, $s\in \mathcal{S}$, corresponding to the $k$-minimum gradient coordinates to update the current state point $t$ to $(1-\alpha)t+ \alpha s$. The step size $\alpha$ in the original FW algorithm is specified by an exogenous rule of diminishing step size \citep{FW1956algorithm}. Since the value function is not convex in general, the exogenous diminishing step size is not guaranteed to enjoy the same convergence behaviours as in convex optimizations \cite{Jaggi2013}. In addition, the FW algorithm suffers from zig-zagging behaviours near the boundary, which weakens the overall performance. This issue can be addressed with the away-steps FW \citep{LacosteJulien2015}, but we avoid the issues altogether by choosing a small constant learning rate with an early stopping criterion since we are more interested in the distance of $t$ to all feasible vertices $s\in \mathcal{S}$ than the exact minimizer $t^*$. In the previous section, we have shown that the negativity of gradient elements of the value function suggests marginal improvement (reduction) of the objective value, and the desirability of including the corresponding asset in the optimal subset portfolio. As such, we can regard the relaxed decision variable $t\in \mathcal{T}$ as the attention towards all assets with a total attention budget specified by the sparsity level $k$. At each FW step, the algorithm checks the current gradient of the value function and picks $k$ coordinates with the most negative gradients, represented by a vertex $s\in \mathcal{S}$ with $1$'s at the corresponding coordinates, and update towards them, $t_i^{(k+1)}=(1-\alpha)t_i^{(k)}+\alpha $. This effectively raises the attention, $t_i$, on the associated assets so that they are more likely to remain active in the subproblem on which Gurobi is applied. For the unselected assets, their assigned attention undergoes an exponential decay where $t_i^{(k+1)}=(1-\alpha)t^{(k)}_i$ so that $\alpha$ also serves as a decay constant. If an asset is repeatedly ignored, it is a strong signal that it has less preferred attributes in the asset universe and it will be less likely to be selected since, by Theorem \ref{theo4}, the gradient component $[\nabla_t f^*]_i$ decays to $0$ as $t_i$ converges to $0$, suggesting no significant marginal improvement of the value function along this coordinate. Consequently, instead of tracking the minimizer $t^*$ to a high precision level, we terminate the iterations early by discarding any assets whose assigned attention value falls below half of the initial uniform attention $\frac{k}{p}$, leading to a fixed number of iterations equals to the smallest integer larger than the half-life $\frac{\ln{0.5}}{\ln (1-\alpha)}$ with decay constant $\alpha$.

We have provided an intuitive account of the FW gradient descent scheme with the half-life early termination criterion. Different from deterministic selection decisions made by $s\in \mathcal{S}$, its relaxed counterpart $t\in \mathcal{T}$ makes fractional decisions, allowing refinement and confirmation iteratively in FWF. Figure \ref{fig:fw_evolution} illustrates this result in a synthetic 100-choose-10 problem generated from a factor model that will be introduced in the Numerical Result section. The figure documents the gradient components and the assigned attention of each assets in Panel (A) and (B) respectively where the 10 assets in the global optimal subset portfolio are colored in red (identified exactly by Gurobi) and other suboptimal assets are colored in green. There are two distinct groups of assets: the first group is discarded assets whose associated gradient components converging to zero and the attention on them decays towards zero as the iteration continues; the second group is reserved assets whose gradient components converges to a uniform negative value and their assigned attention converges to nonzero values that add up to the total attention level $k$ with zero attention on the discarded assets at the limit.

Panel (B) demonstrates that the early iterations in the FW algorithm dominate the categorization of assets into desired and undesired assets since the gradient components of the rarely selected assets in the initial iterations converges to $0$ as their assigned attention decays towards zero. The half-life of decay constant $\alpha$ serves as a clear division of the distinct convergent behaviours of the two asset groups. Appendix \ref{apex2:FWevo} also showcases this categorization scheme when the sample covariance matrix is not well-conditioned, at the level of $\kappa= 1e3$ and $\kappa =1e6$ respectively. Though the scheme does not correctly identify the 10 global optimal assets as desired assets to keep, it does not present significant impairment of the solution quality due to the existence of comparable solution portfolios formed with other suboptimal assets in the reserved group.

We now proceed to provide a pseudo-code of this gradient descent scheme in Algorithm \ref{alg1}, which we refer to as FWF (Frank-Wolfe Filter). In the pseudo-code, to simplify the notations, we use the object $P$ to store all the problem-specific quantities in program \ref{P1.4}. That is $P(\Sigma,\mu, l,u,k,t,\delta)$.  Solving $P$ means solving a subproblem \ref{P1.4.1}, a convex quadratic program parametrized by $t$.

FWF is the main part of the DASH method where a subset of assets is selected, on which Gurobi, or other B\&B solvers, can be applied as usual. By pruning the less preferred assets given a $k$-sparsity constraint - bounded total attention, Gurobi only develops a binary tree on this reserved subset, significantly reducing the complexity by suppressing the combinatorial nature to a lower dimensional problem. In addition, while Theorem \ref{theo3} gives the exact gradient expression of the value function, evaluating it requires first solving a QP subproblem \ref{P1.4.1} at the given $t\in \mathcal{T}$ to determine the active inequality constraints, which is computationally costly in the iteration scheme. To improve the efficiency, we temporarily ignore the box inequality constraints, making the problem an equality constrained parametrized minimization, which admits a closed-form gradient expression  and allows us to skip solving QPs numerically. This substitution does not introduce significant approximation error when a loose box constraint is imposed such that most of the bounds are slacking. On the other hand, if a tight box constraint is present, the conflict of assets not realizing the same efficacy with and without the box constraints is also of limited concern since FWF selects assets with inherently favourable financial attributes and the second-stage Gurobi solves the full constrained problem exactly. The fundamental trade-off is between the filtering accuracy and the filtering speed. We prioritise the latter since the substitution achieves a consistent $5-6\times$ speedup over numerically solving QPs iteratively (Appendix \ref{apex2:fwf}), and the second-stage Gurobi solves the full constrained problem exactly.

Equipped with FWF, DASH is simply as a two-stage algorithm where it firstly selects a subproblem via FWF, then applies Gurobi as usual. We provide a pseudocode of the DASH method in Algorithm \ref{alg2}.

\begin{algorithm}[h]
    \caption{DASH($\Sigma,\mu, l,u,k,\alpha$)}\label{alg2}
    \begin{algorithmic}
    \Require{$\Sigma\in \mathbb{R}^{p\times p}$; $\mu,l,u\in\mathbb{R}^p$, $l\prec0\prec u$; $k\in\mathbb{N}$; $\alpha,\tau \in\mathbb{R}$}
    \Ensure{$s^*\in \mathbb{R}^p$; $f^*\in\mathbb{R}$}

        \State $\mathcal{S}\gets \{1,2,\dots,p\}$
        \State $\mathcal{K} \gets \{\}$

    \For{$i\in \{1,2\}$} \Comment{primary reduction achieved in the first iteration}
        \State $p\gets \Sigma.\text{shape}[0]$
        \State $t\gets k/p\cdot \mathbf{1}$
        \State $\delta \gets
        0.01\min \{\lambda_{\min}(\Sigma),\min\{\text{diag}(\Sigma)\}\}$
        \State $\tau \gets 0.5k/p$ \Comment{half-life elimination threshold}
        \State $P \gets (\Sigma, \mu,l,u,k,t, \delta)$
        \State $t\gets \text{FWF}(\Sigma,\mu, l,u,k,\delta ,\alpha,\tau )$
        \State $\mathcal{K}\gets \{i\in \{1,\dots,p\}:t_i\leq \tau \}$
        \State $\mathcal{S}\gets S\backslash\mathcal{K}$
        \State $\Sigma,\mu,l,u \gets \Sigma(\mathcal{S}),\mu(\mathcal{S}),l(\mathcal{S}),u(\mathcal{S})$ \Comment{extract sub-objects}
    \EndFor

    \State $s,f^*\gets \text{Gurobi}(P)$
    \Comment{solve $P$ as a MIP as formulated in \ref{P1.3}}
    \State $s^*\gets \mathcal{S}(s)$\Comment{identify the binary vector $s^*$ in the original space}
    \State \Return{$s^*,f^*$}
    \end{algorithmic}
\end{algorithm}

\subsection{Abstraction and Generalization to Other Applications}
The proposed DASH method is developed on sparse portfolio selection, but the underlying structure is not unique to finance. Whenever a practical problem reduces to selecting a fixed-size subset from a larger universe of candidates and the quality of each subset can be evaluated by solving a QP as \ref{P1.3.1}, the problem admits the MIQP formulation as \ref{P1.3} and the same relaxation \ref{P1.4}, gradient analysis, and filtering logic applies. To make this generalization precise, we introduce an abstract characterization of this class of problems, reemphasize the difficulty of exponentially growing discrete feasible points, and position DASH as a general metaheuristic for reducing the combinatorial search space prior to applying B\&B algorithms.

\begin{definition}[Subset Selection Space]
    A subset selection space is a tuple $(U,f^*)$ where:
    \begin{enumerate}
        \item $U=\{1,\dots, p\}$ is the universe with $p$ candidates,
        \item Each subset $s\in \{0,1\}^p\backslash \{\mathbf{0}\}$ identifies a candidate subset of $U$, and the collection of all $k$-sparse subsets is
        \[\mathcal{P}^{(k)}=\{s\in\{0,1\}^p\backslash\{\mathbf{0}\}:\|s\|_1=k,\;\;k=1,\dots ,p\}\]
        with $\mathcal{P}=\cup_{k=1}^p \mathcal{P}^{(k)}$ denoting the collection of all non-trivial subsets
        \item $f^*:\mathcal{P}\rightarrow \mathbb{R}$ is the objective function evaluating the quality of each subset.
    \end{enumerate}
\end{definition}

As such, the $k$-sparse subset selection problem is to find the best subset $s^*=\underset{s\in \mathcal{P}^{(k)}}{\text{argmin}}\;f^*(s)$. For a fixed $k$-sparse subset in $\mathcal{P}^{(k)}$, the computation of $f^*(s)$ is easy as the convex QP subproblem \ref{P1.3.1} over the continuous decision variables on the selected candidates, which has a unique solution under the conditions of Lemma \ref{lem1}. The cardinality of the collection of all $k$-sparse subsets is $|\mathcal{P}^{(k)}|= \binom{p}{k}$, $1\leq k\leq p$. This combinatorial nature of the feasible set, $\mathcal{P}^{(k)}$, is the source of NP-hardness of subset selection problems. Though B\&B algorithms can rule out poor feasible points in chunks by pruning infeasible branches while developing a binary decision tree, the exponential growth of the tree remains an inescapable obstacle for efficiency in improving the quality of early stage solutions.

DASH shifts the attention from looking for $k$ candidates in the full candidate universe $U$ to discarding candidates that are least likely to be in the best subset. To formalize the idea, we define the selection process as filters:
\begin{definition}[Filter]
    Given a subset selection space $(U,f^*)$, a filter $\mathcal{F}: 2^{U}\rightarrow 2^U$ is a map such that $\forall S\in 2^U$, $\mathcal{F}(S)\subseteq S$
\end{definition}
In essence, DASH implements a filter $\mathcal{F}$ to the entire candidate universe $U$ to restrict the scope of the problem on a subset $\mathcal{F}(U)$ before applying a B\&B based algorithm. The main challenge faced by DASH, and filters in general, is the trade-off between the improvement of incumbent solutions with limited time budget, and potential exclusion of good candidates from $\mathcal{F}(U)$. We characterize this trade-off through a preference relation on $\mathcal{P}$:
\begin{definition}[Preference]
    Given a subset selection space $(U,f^*)$, the preference of all subsets in $\mathcal{P}$ is a binary relation $R$ on $\mathcal{P}$ such that, for any $a,b\in \mathcal{P}$, $aR b$ if and only if $f^*(a)\leq f^*(b)$
\end{definition}

\begin{lemma}\label{lem2}
    Given a subset selection space $(U,f^*)$, we have the following results
    \begin{enumerate}
        \item R is a total preorder on $\mathcal{P}$.
        \item If $\forall a,b\in \mathcal{P}$ $a\neq b\implies f^*(a)\neq f^*(b)$, then $(\mathcal{P},R)$ is totally ordered.
    \end{enumerate}
\end{lemma}

\begin{proof}
    See Appendix \ref{apex1:proof_lem2}.
\end{proof}
In particular, since $\mathcal{P}^{(k)}\subseteq \mathcal{P}$, Lemma \ref{lem2} also applies to $\mathcal{P}^{(k)}$. The additional assumption in the lemma is satisfied in general as distinct subsets will almost surely yield different objective values. As such, the total order of R on $\mathcal{P}^{(k)}$ is simply restating the existence and uniqueness of a best subset in $\mathcal{P}^{(k)}$.

The formulation \ref{P1.4} provides a natural abstract template. A subset selection problem $(U,f^*)$ falls within the scope of DASH if it satisfies three structural conditions: (i) the subproblem for a fixed subset $s\in \mathcal{P}^{(k)}$ is a convex QP; (ii) the feasibility constraints are linear with box constraints ensuring LICQ; and (iii) the value function $f^*$ is differentiable in the relaxed attention variable $t\in \mathcal{T}$, so that the gradient expression in Theorem \ref{theo3} is well-defined. These conditions are jointly sufficient for the FWF filtering stage to operate without modification.

A canonical application apart from sparse portfolio selection is the subset feature selection in linear regression \citep{Hazimeh2020}, the $\ell_0$-constrained least squares problem
\[\min_\beta \|y-X\beta\|_2^2 \quad \text{subject to} \quad \|\beta\|_0=k.\]
This maps directly onto \ref{P1.4} by substituting $\Sigma \gets X^\top X$, and $\mu\gets X^\top y$, with optional box constraints on $\beta$. Two advantages of DASH carry over verbatim: the solution objective values of all candidate subset features remain intact, and the target sparsity level $k$ is controlled directly rather than implicitly through regularization terms like Lasso or Ridge regression.

\section{Numerical Results}\label{sec:numerical}
In this section, we first examine the Gurobi performance under different problem configurations to understand the relation between the problem difficulty and the related problem parameters. Using the factor model to generate synthetic problems with targeted parameters, a set of numerical experiments is then conducted to compare the performance of DASH and Gurobi. Throughout the comparison, we fix the number of threads used to be 1 for a fair comparison. It is important to note that the numerical results can be extended to multi-thread settings since DASH also runs Gurobi after a small proportion of filtering stage at the start.

\subsection{Synthetic Data Based On Factor Models}
We use a factor model \citep{Ross1976,BaiShi2011} to generate synthetic covariance matrices so that we have direct control over the condition number of the synthetic covariance matrices.

The return vector in a factor model can be described by some common market factors,
\[\mathbf{R} = B f + \epsilon \]
where $B\in \mathbb{R}^{n\times r}$ $(r<p)$, and we assume $f\sim N(0, I_r)$ is the vector of $r$ independent factors and a uniform white noise $\epsilon\in \mathcal{N}(0,\sigma^2 I_r)$ so that we have direct control of the condition number of the generated covariance matrices. Assuming that the factors are uncorrelated with the noise term, we can generate synthetic covariance matrices via \citep{bai2002determining}
\[\Sigma = \mathbb{E}[\mathbf{R} \mathbf{R}^\top]= BB^\top + \sigma^2 I,\]
where $BB^\top = U\Lambda U^\top$ given the singular value decomposition of $B$, $B=U\Lambda^{\frac{1}{2}}V$. After specifying a white noise level $\sigma^2$, i.e. the smallest eigenvalue, and the desired level of condition number of the synthetic covariance matrix, $\kappa$, we can decide the largest eigenvalue based on $cond(\Sigma)= \frac{\lambda_{\max} }{\lambda_{\min}}=\kappa$. The eigenvalues associated with the other factors are generated through a uniform distribution $U(0.5 \lambda_{\max}, \lambda_{\max})$, where $0.5 \lambda_{\max}$ is a discretionary choice of the floor to distinguish eigenvalues corresponding to the market factors from the noises in the market.

Since the expected return vector in the factor model is $0$, $\mathbb{E}[\mathbf{R}]=0$, we generate a non-trivial return vector by assuming a uniform risk premium, $\beta$, for simplicity, so that the synthetic return vector can be computed as
\[\mu = \beta diag(\Sigma) + r_{\epsilon}\]
where $diag(\Sigma)\in \mathbb{R}^p$ is the diagonal vector of asset variances and $r_\epsilon$ is a vector of independent return noises proportional to the asset variances, $[r_\epsilon]_i \in \mathcal{N}(0,0.05\Sigma_{ii})$.

\subsection{Gurobi Performance Under Different Problem Configurations}
\begin{figure}[h]
    \centering
    \makebox[\textwidth][c]{\includegraphics[scale=0.45]{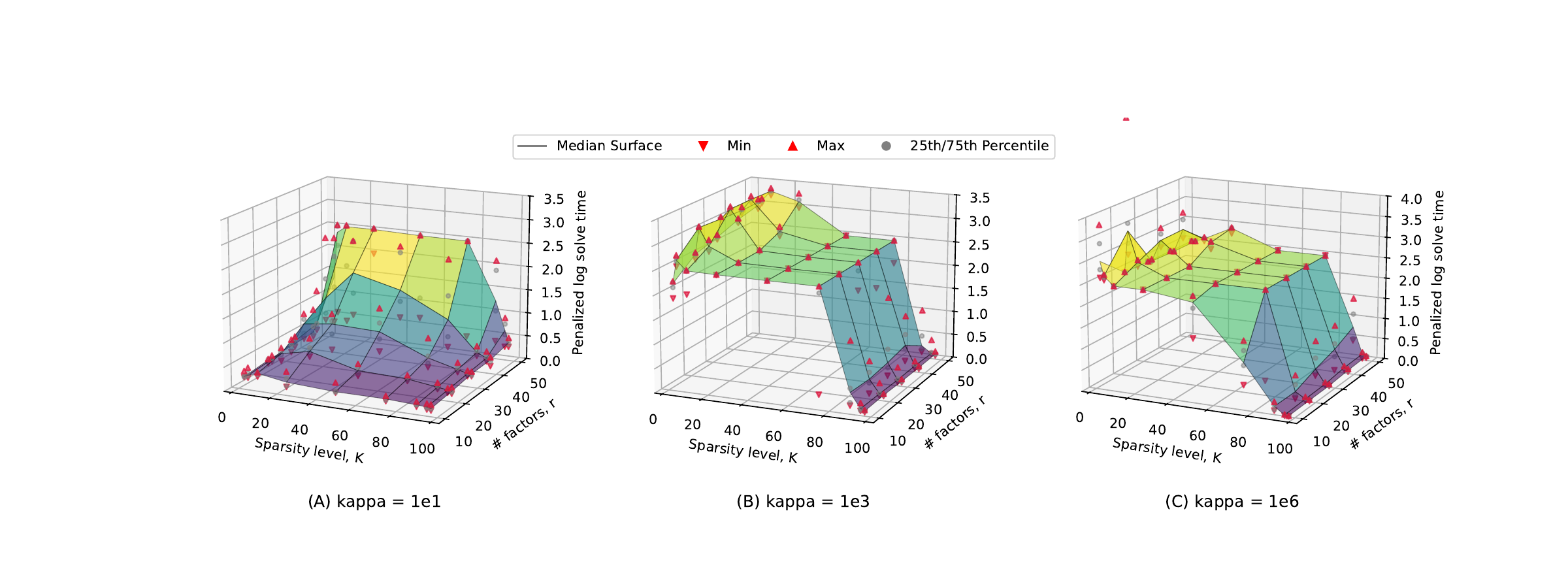}}
    \caption{Difficulty of A Spectrum of Problem Configurations, $\kappa$, $K$, $r$}
    \label{fig:gurobi_results}
    \raggedright
    \footnotesize\textit{Note: For each parameter bundle $(\kappa,k,r)$, 10 sample problems are solved by Gurobi with a time-budget of 300s. The performance is measured by PST (Penalized Solve Time), $PST = time\times (scaler+1)$ where $scaler = \frac{\text{incumbent obj. val. - lower bound}}{\max\{|\text{incumbent obj. val.}|,1\}}$}
\end{figure}

The performance of Gurobi differs significantly in different problem configurations, suggesting the varying difficulty of the problems. It is clear that the solve time grows exponentially with the problem size ($p$). Other problem attributes of interest that affect the problem difficulty and therefore, the performance of Gurobi are the condition number of the covariance matrix $(\kappa)$, the sparsity level ($k$, relative to $p$), and the number of factors used to generate the synthetic covariance matrix. In Figure \ref{fig:gurobi_results}, we fix $p=100$ and estimate the performance of GUROBI with different configurations of $\kappa$, $k$, and $r$. We consider the grid spaces:
\[k\in [3,5,10,25,50,75,90,95,97],\]
\[\kappa \in [1e1, 1e3,1e6],\]
\[r\in [10,20,30,40,50].\]

The condition number of the covariance matrix influences the range of the landscape of the value function in that the spectrum of eigenvalues of an ill-conditioned covariance matrix is wider than that of a well-conditioned matrix, leading to more dispersion of the feasible solutions so that branch-pruning is less efficient in B\&B and the binary tree tend to be developed in full.

The sparsity constraint contributes to the complexity of the problem by determining the size of the search space, the number of feasible points. A naive hypothesis is that the complexity, in terms of solve time, is highest and symmetric at $k=\frac{p}{2}$. However, based on Figure \ref{fig:gurobi_results}, the difficulty rises as the sparsity constraint decreases further before the combinatorial number of the associated search space reduces to $p$ at $k=1$. This provides a signal that the difficulty of an MIP is asymmetric in the sparsity level, so that it is easier to discard a few least preferred assets than to pick a non-trivial subset of the most preferred assets. This asymmetric characteristic provides additional justification of the appropriateness of DASH in discarding undesired assets to focus on a subset of assets.

Lastly, the number of factors in the factor model affects the complexity of the value function landscape by increasing the rank of the synthetic covariance matrix prior to adding a uniform noise. The covariance matrix has more complicated structure when generated from more factors. As shown in Panel (A), there is a clear monotonic trend of increasing problem complexity with the number of factors in well-conditioned covariance matrices, $\kappa = 1e1$. For moderately and severely ill-conditioned matrices in Panel (B) and (C), the trend is only marginal or none, suggesting the more dominant role played by the condition number in affecting the problem complexity over the number of factors. As such, we set the number of factors to be $10\%$ of the problem size in the following numerical experiments, which shows good Gurobi performance in well-conditioned covariance matrices $(\kappa = 1e1)$ and ensures self-similarity of the covariance structure across problem sizes, preventing unaccounted impact of inconsistent factor-to-asset ratio.

We have identified that the condition number and the sparsity level are of dominant impact on the difficulty of the problem. As such, the performance comparison will be conducted by varying these parameters for robustness.

\subsection{Performance Comparison Between DASH and Gurobi}
\vspace{-5pt}
\begin{table}[h!]
\centering
\footnotesize
\setlength{\tabcolsep}{3pt}
\renewcommand{\arraystretch}{1.3}
\caption{Performance Comparison: $\kappa=1e6$, $[l,u]=[-0.3,1]$}
\label{table:5.1}
\begin{adjustbox}{max width = \textwidth}
\begin{threeparttable}
\begin{tabular}{ccrrrrrrrrrrrr}
\toprule
 &  & \multicolumn{3}{c}{Time (60s)} & \multicolumn{3}{c}{Time (180s)} & \multicolumn{3}{c}{Time (300s)} & \multicolumn{3}{c}{Time (600s)} \\
p & k & DASH & Gurobi & Win\% & DASH & Gurobi & Win\% & DASH & Gurobi & Win\% & DASH & Gurobi & Win\% \\
\midrule
\multirow{3}{*}{100} & 3 (3.0\%) & -0.6993 & -1.1311 & 0.0\% & -0.6993 & -1.1311 & 0.0\% & -0.6993 & -1.1311 & 0.0\% & -0.6993 & -1.1311 & 0.0\% \\
 & 5 (5.0\%) & -1.4599 & -1.6943 & 0.0\% & -1.4599 & -1.6943 & 0.0\% & -1.4599 & -1.6943 & 0.0\% & -1.4599 & -1.7008 & 0.0\% \\
 & 10 (10.0\%) & -2.9386 & -3.0002 & 10.0\% & -2.9386 & -3.0002 & 10.0\% & -2.9386 & -3.0002 & 10.0\% & -2.9386 & -3.0002 & 10.0\% \\
\hline
\multirow{3}{*}{500} & 5 (1.0\%) & -0.7632 & -0.9081 & 0.0\% & -0.7632 & -0.9186 & 0.0\% & -0.7632 & -0.9327 & 0.0\% & -0.7632 & -0.9585 & 0.0\% \\
 & 25 (5.0\%) & -3.3232 & -3.4314 & 10.0\% & -3.3471 & -3.4855 & 10.0\% & -3.3534 & -3.5005 & 10.0\% & -3.3565 & -3.5089 & 10.0\% \\
 & 50 (10.0\%) & -7.2344 & -7.2381 & 30.0\% & -7.2415 & -7.2701 & 30.0\% & -7.2415 & -7.2758 & 20.0\% & -7.2573 & -7.2927 & 40.0\% \\
\hline
\multirow{3}{*}{1500} & 15 (1.0\%) & -1.2448 & -1.2709 & 10.0\% & -1.2452 & -1.2741 & 0.0\% & -1.2452 & -1.2776 & 0.0\% & -1.2463 & -1.2787 & 0.0\% \\
 & 75 (5.0\%) & -6.6612 & -6.5598 & 60.0\% & -6.6910 & -6.9133 & 30.0\% & -6.6910 & -6.9909 & 30.0\% & -6.6966 & -7.0469 & 20.0\% \\
 & 150 (10.0\%) & -16.5576 & -15.9982 & 100.0\% & -16.7574 & -16.4987 & 80.0\% & -16.7688 & -16.6516 & 70.0\% & -16.8196 & -16.7361 & 60.0\% \\
\hline
\multirow{3}{*}{3000} & 30 (1.0\%) & -1.4552 & $6.6122^\dagger$ & 100.0\% & -1.4599 & -1.4348 & 80.0\% & -1.4599 & -1.4579 & 50.0\% & -1.4657 & -1.4671 & 50.0\% \\
 & 150 (5.0\%) & -10.1881 & $6.6122^\dagger$ & 100.0\% & -11.2634 & -6.6797 & 100.0\% & -11.3691 & -10.0102 & 100.0\% & -11.4288 & -10.8723 & 100.0\% \\
 & 300 (10.0\%) & -17.1251 & $6.6122^\dagger$ & 100.0\% & -28.3631 & -13.5538 & 100.0\% & -29.1356 & -27.4673 & 100.0\% & -29.6070 & -28.8596 & 90.0\% \\
\bottomrule
\end{tabular}
\begin{tablenotes}
\footnotesize
\item[*] DASH and Gurobi values report the average incumbent objective across 10 synthetic problems.
Win\% is the percentage of problems on which DASH achieves a lower or equal incumbent objective value than Gurobi at the given time budget.
\item[$\dagger$] Gurobi reports the initial heuristic solutions with no better incumbent solutions. All such instances are counted as wins for DASH.
\end{tablenotes}
\end{threeparttable}
\end{adjustbox}
\end{table}
\vspace{-5pt}
Table \ref{table:5.1} shows a comprehensive performance examination of DASH and Gurobi in terms of the solution quality over time across different problem sizes and sparsity levels by comparing the average solution quality and the win rate of DASH over Gurobi. The algorithms are run on 10 synthetic problems in each problem size $p\in [100,500,1500,3000]$, and we consider three levels of sparsity $k\in[\min\{3,1\%p\}, 5\% p, 10\%p\}]$ and a time budget $T\in [60, 180, 300, 600]$ in seconds. We choose $\kappa =1e6$ and $[l,u]=[-0.3, 1]$ to mimic the situation where the sample covariance matrix is ill-conditioned and a tight box constraint is imposed to prevent extreme weights on single assets. It is noted that for simple small-scale problems like $p=100$, Gurobi tends to find global optimal solutions before the specified time limit. We do not make a distinction of the convergent cases separately since the focus is mainly on the quality of incumbent solutions for more difficult large-scale problems under a given time budget.

For $p=100$, the gap of average solution quality between DASH and Gurobi remain fixed across time. This is largely due to the incorrect exclusion of optimal assets by DASH when the covariance matrix is ill-conditioned, as shown in Appendix \ref{apex2:FWevo}. As the problem size increases, the relative gap narrows and DASH eventually outperforms Gurobi with more prominent improvement of incumbent solutions over most of the sample problems. At $p=1500$, the two methods achieve broadly comparable solution quality across all time budgets. At $p=3000$, DASH shows the most pronounced advantage in the early stage, outperforming Gurobi in all instances at $T=60s$, before Gurobi progressively narrows the gap as the time budget grows. The result is consistent for well-conditioned ($\kappa =1e1$) and moderately-conditioned ($\kappa = 1e3$) matrices, which is shown in Appendix \ref{apex2:GenPerCom}. Additionally, as the covariance matrix becomes well-conditioned, the relative improvement of DASH over Gurobi weakens due to the decrease of problem difficulty as observed in Figure \ref{fig:gurobi_results}. The relation between the condition number of the covariance matrix and the computational difficulty of the problem is further demonstrated in the following section where we examine the robustness of DASH's improvement over Gurobi with varying problem configurations.

\subsection{Sensitivity Analysis}
\begin{figure}[tbp]
    \centering
    \includegraphics[width=0.6\linewidth]{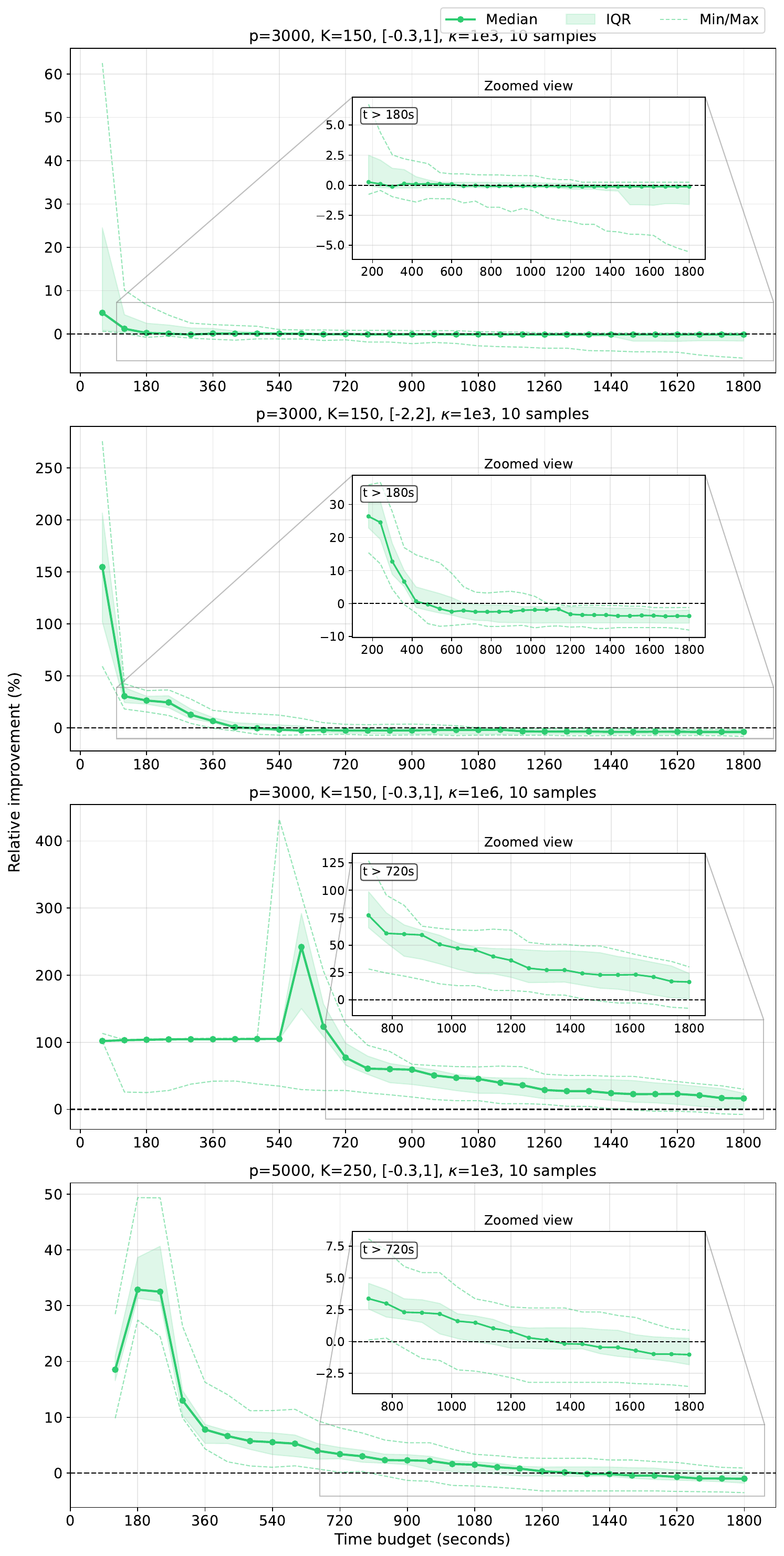}
    \caption{Sensitivity Analysis of DASH Improvement}
    \label{fig:sensitivity analysis}
\end{figure}

As we have shown in the previous subsection, in addition to the combinatorial nature of the problem size, the difficulty of the problem is also related to the condition number of the covariance matrix. Consequently, the box constraint is also of crucial importance due to its shrinkage-like effect on the problem \citep{JagannathanMa2003}. In Figure \ref{fig:sensitivity analysis}, we show the sensitivity analysis of the performance comparison between DASH and Gurobi under different problem configurations. Denote the incumbent solution of DASH and Gurobi at time $t$ by $f_d(t)$ and $f_g(t)$ respectively, the figure tracks the relative gap
\[R_g=\frac{f_g(t)-f_d(t)}{|f_g(t)|}\]
over a time span of 1800 seconds. $R_g$ measures the percentage improvement or degradation of the incumbent solution quality of DASH over that of Gurobi. We set $p=3000$, $K=150$, $[l,u]=[-0.3,1]$ as the benchmark case in the first subplot where the improvement over Gurobi is short-lived and Gurobi can achieve better solution as the time budget increases. In the remaining three subplots, the box constraint, the condition number, and the problem size are perturbed separately to increase the problem difficulty marginally and to examine the robustness of the improvement of DASH over Gurobi. The improvement is over two dimensions: the relative improvement for a fixed time budget and the duration of time of improvement over Gurobi. When relaxing the box constraint to $[-2,2]$ as in subplot 2, an expansion of the improvement region is clear in both the vertical magnitude and the horizontal duration where an initial median improvement of 150\% is achieved and more than 10\% improvement is persistent in the first 300 seconds before gradually decaying towards 0. Similar vertical and horizontal expansion of the improvement region is also observed in subplots 3 and 4 where we consider the ill-conditioned matrices and a larger problem size respectively. The horizontal improvement is consistent in both cases while the vertical improvement is significantly more substantial for ill-conditioned matrices.

Overall, the sensitivity analysis confirms that DASH delivers significant improvement of the incumbent solutions precisely when the problem is most difficult, which can be a loose box constraint, an ill-conditioned covariance matrix, and a large problem size. This highlights the capacity of DASH in making significant contribution to difficult problem settings in practice. In such settings, the filtering stage of DASH effectively reduces the combinatorial search space to a subspace where Gurobi can use computation resources more efficiently to find high-quality incumbent solutions.

\section{Conclusion}\label{sec:conclusion}
This paper presents DASH, a dimensionality reduction method that facilitates commercial MIP solvers such as Gurobi in finding better incumbent solutions of difficult subset portfolio selection problems, where the difficulty can arise from the problem size, the condition number of the sample covariance matrix, and the box constraint.

DASH suppresses the combinatorial nature of the problem by identifying a subproblem on which Gurobi is applied. The filtering of assets risks deterioration of solution quality in exchange for saving computational resources. We compare the performance of DASH and direct application of Gurobi on problems with varying configurations. The numerical results collectively show that significant improvement of incumbent solution quality is achieved for difficult problems that are of practical relevance. For example, if a practitioner respects an ill-conditioned sample covariance matrix but worries about the reliability of extreme solver solutions, instead of applying a shrinkage method to the matrix, a proper box constraint can be applied to directly prevent the solver from placing extreme weights on assets.

For tractable problems with a small problem size, well-conditioned covariance matrix, and tight box constraints — DASH offers little improvement over Gurobi. Based on the extensive set of synthetic experiments, DASH offers significant improvement when the problem size is large, the covariance matrix is ill-conditioned, and the box constraint is loose. When a combination of the conditions is present, the improvement of DASH over Gurobi is most pronounced and persistent.

The primary limitation of DASH is the possibility of incorrect exclusion of optimal assets in the filtering stage. This concern is of limited practical importance since comparable high-quality incumbent solutions tend to exist in the filtered subset and DASH can be complementary to Gurobi by using the better-quality incumbent solution of DASH as an upper bound for Gurobi, enabling more efficient branch pruning and reducing the overall computational burden.

\bibliography{references}

\appendix
\section{Proofs}

\subsection{Proof of Lemma \ref{lem1}}\label{apex1:proof_lem1}

\begin{proof}
    \begin{enumerate}
        \item The Hessian of the objective function is $\tilde{\Sigma}_t$. For any arbitrary $x\in \mathcal{B}$, we have that
        \begin{align*}
            x^\top \tilde{\Sigma}_t x & = x^\top T_t \Sigma T_t x + \delta (\|x\|_2^2 -\|T_tx\|_2^2 ) \\
            &\geq (T_tx)^\top \Sigma (T_tx)+\delta(1-\|T_t\|^2)\|x\|_2^2
        \end{align*}
        If $x\notin \ker T_t$, then $(T_tx)^\top \Sigma (T_tx)\geq0$, where we have strict inequality if $\Sigma$ is positive definite. Also, $\delta (1-\|T_t\|^2)\|x\|_2^2\geq 0$ since $0<\max t\leq 1$ $(\mathbf{0}\notin \mathcal{T})$ and $0<\|T_t\|\leq 1$, where we can achieve strict inequality if $t\prec 1$ so that $\|T\|<1$ and $\delta(1-\|T_t\|^2)\|x\|_2^2>0$, ignoring the infeasible point $x=\mathbf{0}\notin X(t)$. Hence, $x^\top \tilde{\Sigma}_t x\geq 0$ where we have strict inequality if $\Sigma$ is positive definite or $0\prec t\prec 1$.\\
        On the other hand, if $x\in \ker T_t $, then $(T_tx)^\top \Sigma (T_tx)=0$, and $\delta (\|x\|_2^2-\|T_tx\|_2^2)=\delta \|x\|_2^2$, so that $x^\top \tilde{\Sigma}_t x=\delta \|x\|_2^2\geq 0$ where equality is only achieved at the infeasible point $x=\mathbf{0}$.\\
        As such, $\tilde{\Sigma}_t$ is positive semi-definite and $f$ is convex. If $\Sigma$ is positive definite or $t\prec \mathbf{1}$, then $\tilde{\Sigma}_t$ is positive definite and $f$ is strictly convex.\\
        It is clear that the feasible set is convex with linear equality and inequality constraints. Hence, \ref{P1.4.1} is a convex optimization problem. If $\Sigma$ is positive definite or $t\prec \mathbf{1}$, then \ref{P1.4.1} is a strict convex optimization problem.\\
        To see $x^*_i = 0$ when $t_i = 0$, simply note that $x_i$ is decoupled from the other coordinates and the linear term vanishes. That is, $\text{argmin}_{x_i}\delta x_i^2 = 0$.
        \item Since $G=[-I,I]^\top$ and $b=[-l,u]^\top$ with $l<0<u$, the active inequality constraints at any $x\in X(t)$ are of the form $l_i-t_ix_i=0$ or $t_ix_i-u_i=0$ for some $i\in\{1,\dots,p\} $, and we denote the all bounded coordinates by $\mathcal{I}(x):=\{i\in \{1,\dots, p\}: l_i-t_ix_i=0\;or\; t_ix_i-u_i=0\}\}$ as defined in Assumption \ref{assp1}. The corresponding inequality gradients are of the form $\pm t_ie_i$. Since the lower and upper bounds cannot be active simultaneously, the active inequality constraint gradients are nonzero multiples of distinct basis vectors $e_i$, $i\in \mathcal{I}(x)$, so that they are linearly independent.
        The gradient of the aggregate-weight equality constraint is $t=\sum_i t_ie_i=\sum_{i\in \text{supp}(t)}t_ie_i$. By Assumption \ref{assp1}, since there exists $j\in \text{supp}(t)\backslash \mathcal{I}(x)$, $t$ cannot be expressed as a linear combination of $e_i$, $i\in\mathcal{I}(x)$. As such, $t$ is also linearly independent to the inequality constraint gradients. Hence, we conclude that LICQ holds for all $x\in X(t)$.
        \item From the previous two results, since LICQ holds for all $x\in X(t)$ and $x^*(t):=\underset{x\in X(t)}{\text{argmin}} f(x;t)$ solves convex program \ref{P1.4.1}, we have KKT conditions at $x^*(t)$.
    \end{enumerate}

\end{proof}

\subsection{Proof of Theorem \ref{theo2}}\label{apex1:proof_theo2}
To prove Theorem \ref{theo2}, we need the following two propositions \citep[pp.231-232]{sundaram1996first}:

\begin{prop}\label{prop2}
    Let $X:\mathcal{T}\rightarrow P(\mathcal{B})$ be a compact-valued correspondence. Then, $X$ is upper-hemicontinuous at $t\in \mathcal{T}$ if and only if for all sequences $\{t^{(n)}\}_n\rightarrow t\in \mathcal{T}$, and for all sequences $\{x^{(n)}\}_n$ with $x^{(n)}\in X(t^{(n)})$, there is a subsequence $\{x^{(n_k)}\}_k\subset \{x^{(n)}\}_n$ such that $\{x^{(n_k)}\}_k$ converges to some $x\in X(t)$.
    \end{prop}
\begin{prop}\label{prop3}
    Suppose $X:\mathcal{T}\rightarrow P(\mathcal{B})$ is lsc at $t$, and $x\in X(t)$. Then, for all sequences $t^{(n)}\rightarrow t$, there is $x^{(n)}\in X(t^{(n)})$ such that $x^{(n)}\rightarrow x$.\\
    Conversely, let $X:\mathcal{T}\rightarrow P(\mathcal{B})$, and let $t \in \mathcal{T}$, $x\in X(t)$. Suppose that for all sequences $t^{(n)}\rightarrow t$, there is a subsequence $n_k$ of $n$, and $x^{(n_k)}\in X(t^{(n_k)})$ such that $x^{(n_k)}\rightarrow x$. Then, $X$ is lower-hemicontinuous at $t$.
\end{prop}

\begin{proof}
\begin{enumerate}
    \item The result follows from Berge's Maximum Theorem, or more precisely, the Minimum Theorem for minimization problems as in Theorem \ref{theo1}. As such, we prove the following conditions in the Minimum Theorem:
    \begin{enumerate}
        \item $f$ is a continuous function on $\mathcal{T} \times \mathcal{B}$
        \item $X: \mathcal{T}\rightarrow P(\mathcal{B})$ is a compact-valued continuous correspondence.
    \end{enumerate}
    Firstly, it is clear that $f(t,x)$ is continuous on the product space $\mathcal{T}\times \mathcal{B}$. We now show that the correspondence $X$ is compact-valued and continuous. For an arbitrary $t\in \mathcal{T}$, $X(t)\subset\mathcal{B}$ is closed and bounded. Hence, by Heine-Borel theorem, $X(t)$ is compact so that the correspondence $X$ is compact-valued. We also note that $X(t)$ is convex so that $X$ is convex-valued as well.\\
    To show the correspondence $X$ is continuous, we show that it is both upper-hemicontinuous and lower-hemicontinuous.\\
    For an arbitrary sequence $\{t^{(n)}\}_n \rightarrow t\in \mathcal{T}$ and an arbitrary sequence $\{x^{(n)}\}_n\subset  X(t_n)$, by a proposition above we want to show that there exists a subsequence $\{x^{(n_k)}\}_k\subset \{x^{(n)}\}_n$ such that it converges to some $x\in X(t)$. Since $\mathcal{B}$ is compact, by Bolzano-Weistrass theorem, there exists a subsequence $\{x^{(n_k)}\}_k\subset \{x^{(n)}\}_n$ that converges to some $x$. For $k\in\mathbb{N}$, we have $x^{(n_k)}\in X(t^{n_k})$ so that $GT_{t^{(n_k)}}x^{(n_k)}-b\preceq 0$ and $AT_{t^{(n_k)}}x^{(n_{k})}-c=0$. Since $T_{t^{(n_k)}}\rightarrow T_t$ and $x^{(n_k)}\rightarrow x$ and left-hand sides of inequality and equality constraints are continuous in both $t$ and $x$, we have $GT_tx-b=\lim_k\Big(GT_{t^{(n_k)}}x^{(n_k)}-b \Big)\preceq 0$ and $AT_tx-c = \lim_k\Big(AT_{t^{(n_k)}}x^{(n_k)} -c\Big)=0$ as $(-\infty,0]^m$ and $\{0\}$ are closed. Hence, $x\in X(t)$, so that $X$ is upper-hemicontinuous.\\

    We now proceed to show that $X$ is lower-hemicontinuous. Let $t\in\mathcal{T}$, $x\in X(t)$, for an arbitrary sequence $\{t^{(n)}\}_n \rightarrow t$, we want to show that there is a subsequence $\{t^{(n_k)}\}_k\subset \{t^{(n)}\}_n$ and $\{x^{(n_k)}\}_k\subset \{x^{(n)}\}_n$ where $x^{(n)}\in X(t^{(n)})$ such that $x^{(n_k)}\rightarrow x$. \\
    Denote the index set of active inequality constraints as $\mathcal{A}:=\{i\in \{1,\dots,2p\}: [GT_tx-b]_i=0\}$, $\mathcal{A}^c = \{i\in\{1,\dots,2p\}: [GT_tx-b]_i<0\}$ and $G_{\mathcal{A}}\in \mathbb{R}^{|\mathcal{A}|\times p}$ as a sub-matrix of $G$ that contains the active rows of $G$. Then we have the following linear system:
    \[\underbrace{\begin{bmatrix}
        AT_t\\
        G_{\mathcal{A}}T_t
    \end{bmatrix}}_{=:L_t\in \mathbb{R}^{(|\mathcal{A}|+1)\times p}}x=\begin{bmatrix}
        c\\
        b_{\mathcal{A}}
    \end{bmatrix},
\]
 and the inequality system \[G_{\mathcal{A}^c}T_tx\prec b_{\mathcal{A}^c}\]
where $G_{\mathcal{A}^c}\in\mathbb{R}^{(2p-|\mathcal{A}|)\times p}$. Since $G=[-I, I]^\top$, $b=[l,u]^\top$, the only inequality constraints on $x_i$ are $l_i\leq t_i x_i\leq u_i$. If $t_i=0$, both inequality constraints on $x_i$ are inactive, so the box inequality constraints on $x_i$ can only be active if $t_i>0$. Hence, the active rows of $G_{\mathcal{A}}T_t$ are not trivially all zeros.\\
By LICQ (Lemma \ref{lem1}), since the rows of $L_t$ are gradients of constraints, they are linearly independent. In particular, the gradients of the linear system are $t=\sum_{i\in\text{supp}}t_ie_i$ and $\pm t_ie_i$, where the sign is $+$($-$) if upper(lower) bound is binding, for $i\in \mathcal{I}(x)=\{i\in \{1,\dots,p\}: l_i-t_ix_i=0 \;or\; t_ix_i-u_i=0\}$ as in Assumption \ref{assp1}. We note that $|\mathcal{I}(x)|=|\mathcal{A}|$ since each bounded coordinates $i\in \mathcal{I}(x)$ corresponds to an active inequality constraint $j\in \mathcal{A}$.
Hence, the rows of $L_t$ as continuous functions of $t$ are bounded away from $\mathbf{0}$ so that there exists $\epsilon>0$ such that for all $\tau \in B_\epsilon(t)$, $\tau_i>0\;\forall i\in \text{supp}(t)$ and $L_\tau^\top:=[T_\tau A^\top, T_\tau G_{\mathcal{A}}^\top]^\top$ is non-degenerate, i.e. $rank(L_\tau)=rank(L_t)=|\mathcal{A}|+1$. Since, for an arbitrary sequence $\{t^{(n)}\}_n\rightarrow t$, there exists $N_0$ large enough such that $t^{(n)}\in B_\epsilon(t)$ for all $n\geq N_0$, we have that $\lim_n L_{t^{(n)}}^+=L_t^+$, where $L^+$ is the pseudo-inverse of $L_t$ and its continuity is guaranteed since $rank(L_{t^{(n)}})=rank(L)=|\mathcal{A}|+1$ \citep{Rakocevic1997}.\\
    Hence, given such a sequence $\{t^{(n)}\}_n$, define a sequence $\{\delta^{(n)}\}_n$ such that, for all $n\in\mathbb{N}$, $\delta^{(n)}$ satisfies the following linear system:
    \[
\underbrace{\begin{bmatrix} AT_{t^{(n)}} \\ G_{\mathcal{A}}T_{t^{(n)}} \end{bmatrix}}_{=:L_n} \delta^{(n)} = \underbrace{\begin{bmatrix} c - AT_{t^{(n)}}x \\ b_{\mathcal{A}}- G_{\mathcal{A}}T_{t^{(n)}}x \end{bmatrix}}_{=:r^{(n)}}.
\]
Define $\delta^{(n)}$ as the minimum norm particular solution, $\delta^{(n)}=L_n^+r^{(n)}$.
 By construction, since $rank(L_n)=rank(L)=|\mathcal{A}|+1$, we have $\lim_n L_n^+ =L^+$, so that $\lim_n \delta^{(n)}=L^+\lim_n r^{(n)}=0$. Hence, we define a sequence $\{x^{(n)}\}_n$ with $x^{(n)}=x+\delta^{(n)}$, for all $n\in\mathbb{N}$ and we note that $x^{(n)}\rightarrow x$.\\
 For all $n\in\mathbb{N}$, verify that
 \begin{align*}
     L_nx^{(n)}&=L_n(x+\delta^{(n)})=L_nx+L_nL_n^+(\begin{bmatrix}
         c\\
         b_{\mathcal{A}}
     \end{bmatrix} - L_nx)=\begin{bmatrix}
         c\\
         b_{\mathcal{A}}
     \end{bmatrix}
 \end{align*}
where $L_nL_n^+=I$ since $L_n$ has full row rank, i.e. $rank(L_n)=|\mathcal{A}|+1$ for $n\geq N_0$. Hence, $x^{(n)}$ satisfies the equality constraint and the active subset of inequality constraints by equality. To satisfy the complementary set of inactive inequality constraints, we note that:
\begin{align*}
    G_{\mathcal{A}^c}T_{t^{(n)}}x^{(n)}&=G_{\mathcal{A}^c}(T_{t^{(n)}}-T_t+T_t)(x+\delta^{(n)})\\
    &= G_{\mathcal{A}^c}(T_{t^{(n)}}-T_t)(x+\delta^{(n)})+G_{\mathcal{A}^c}T_t(x+\delta^{(n)})\\
    &\rightarrow 0 + G_{\mathcal{A}^c} T_tx\prec b_{\mathcal{A}^c}
\end{align*}
since $T_{t^{(n)}}\rightarrow T_t$ and $\delta^{(n)}\rightarrow 0$. Hence, there exist $N_1\in\mathbb{N}$ large enough, such that $G_{\mathcal{A}^c}T_{t^{(n)}}x^{(n)}\prec b_{\mathcal{A}^c}$ and $x^{(n)}\in X(t^{(n)})$ for all $n\geq N_1$.
    Thus, upon truncating the first $\max\{N_0,N_1\}$ terms and re-indexing, we have $x^{(n)}\in X(t^{(n)})$ and $x^{(n)}\rightarrow x$ for sufficiently large $n$.
    As such, by Proposition \ref{prop3}, $X$ is also lower-hemicontinuous. Since $X$ is both upper-hemicontinuous and lower-hemicontinuous, $X$ is a continuous correspondence.\\
    We have shown that $f$ is a continuous function on $\mathcal{T}\times \mathcal{B}$, and $X$ is a compact-valued continuous correspondence. In addition, by Lemma \ref{lem1}, $f(\cdot; t)$ is strictly convex in $x$ for each $t\in\mathcal{T}$, and $X(t)$ is clearly convex as a nontrivial intersection of a hypercube with a hyperplane. Thus, by Theorem \ref{theo1}, we can conclude that $f^*$ and $X^*$ are both continuous functions on $\mathcal{T}$, and we can denote the latter by $x^*(t)=X^*(t)$ to emphasize that $X^*$ is singled-valued.\\

    \item Lastly, to show that the dual variables $\lambda^*: \mathcal{T}\rightarrow \mathbb{R}^{2p}, \;\nu^*:\mathcal{T}\rightarrow \mathbb{R}$ are continuous in $t$, for an arbitrary $t\in \mathcal{T}$, we write the stationarity condition in Lemma \ref{lem1} as:
    \[\begin{bmatrix}
        T_tA^\top & T_tG^\top
    \end{bmatrix}\begin{bmatrix}
        \nu^*\\
        \lambda^*
    \end{bmatrix}=2\tilde{\Sigma}_tx^* -T_t\mu,\]
    where the right hand side is a continuous function of $t$ by the continuity of $x^*$ in $t$ proved in part one. Identically define the set of active inequality constraints at $x^*$ by $\mathcal{A}=\{i\in \{1,\dots, 2p\}: [GT_tx^*-b]_i=0\}$
    and we can rearrange and write the stationarity condition as
    \[\begin{bmatrix}
        T_tA^\top & T_tG_{\mathcal{A}}^\top & T_tG_{\mathcal{A}^c}^\top
    \end{bmatrix}\begin{bmatrix}
        \nu^*\\
        \lambda_{\mathcal{A}}^*\\
        \lambda_{\mathcal{A}^c}^*
    \end{bmatrix} = \underbrace{\begin{bmatrix}
        T_tA^\top & T_tG_{\mathcal{A}}^\top
    \end{bmatrix}}_{=L_t^\top}\begin{bmatrix}
        \nu^*\\
        \lambda_{\mathcal{A}}^*
    \end{bmatrix}+0=2\tilde{\Sigma}_tx^* -T_t\mu\]
    where $\lambda^*_i(t)=0$ for all $i\in\mathcal{A}^c$ by complementary slackness (Lemma \ref{lem1}) and $L_t$ is defined identically in part 1 so that there exists $\epsilon>0$ such that $\text{rank}(L_\tau)=\text{rank}(L_t)=|\mathcal{A}|+1$ for all $\tau \in B_\epsilon(t)\cap \mathcal{T}$ and the pseudoinverse $L_t^+$ is continuous at $t$ \citep{Rakocevic1997}. Hence, we have that
    \[\begin{bmatrix}
        \nu^*\\
        \lambda_{\mathcal{A}}^*
    \end{bmatrix}=L_t^+\Big(2\tilde{\Sigma}_tx^*-T_t\mu\Big)\]
    and therefore, $\nu^*$, $\lambda^*_{\mathcal{A}}$ are continuous functions of $t$ as products of two continuous functions.\\
    For $i\in \mathcal{A}^c$, suppose it corresponds to the inactive lower bound on $x_j$ for a $j\in \{1,\dots,p\}$, then we have $l_j<t_jx^*_j(t)$. Since $T_tx^*(t)$ is continuous in $t$, each coordinate $t_jx^*_j(t)$ is continuous in $t$ as well. As such, there exists $\epsilon >0$ such that for all $\tau \in B_\epsilon (t)$, we have $l_j < \tau_j x^*_j(\tau)$ and $\lambda^*_i(\tau)=\lambda^*_i(t)=0$ by the complementary slackness condition. Hence, $\lambda^*_{\mathcal{A}^c}$ is trivially continuous as zeros at $t$. The same argument holds if $i\in \mathcal{A}^c$ corresponds to an inactive upper bound. Hence, $\lambda^*_{\mathcal{A}^c}$ is continuous in $t$, and more generally, $\lambda^*$ is continuous in $t$.
\end{enumerate}
To conclude, we have established that $x^*(t), f^*(t),\lambda^*(t) $, and $\nu^*(t)$ are all continuous on $\mathcal{T}$. This completes the proof of Theorem \ref{theo2}.
\end{proof}

\subsection{Proof of Theorem \ref{theo3}}\label{apex1:proof_theo3}
\begin{proof}
    We adopt the same notation as in \cite{sundaram1996first}. For example, $df$ is the total differential of $f$ and $D_tf\in \mathbb{R}^{1\times p}$ is the partial derivative of $f$ w.r.t. $t\in \mathbb{R}^p$.

    By Theorem \ref{theo2}, $x^*, \lambda^*, \nu^*$ are continuous functions of $t$ in \ref{P1.4}. The value function can be written as:
    \begin{align*}
        \mathcal{L}^*(t)=  L(x^*,\lambda^*,\nu^*;t)& = f(t,x^*(t)) + (\lambda^*)^\top g(t,x^*(t)) + (\nu^*)^\top h(t,x^*(t))
    \end{align*}
    where $f(t,x^*(t))=(x^*)^\top \tilde{\Sigma}_t x^* -\mu^\top T_tx^*$, $g(t,x^*(t))=GT_tx^*-b$ and $h(t,x^*(t))=AT_tx^*-c$. By complementary slackness and primal feasibility of KKT conditions (Lemma \ref{lem1}) we note that the Lagrangian estimated at the optimum is equal to the value function, $\mathcal{L}^*(t)=f(t,x^*(t))$. Noting the dependency of $x^*$ on $t$, we suppress the notations as $f^*(t)$, $g^*(t)$, and $h^*(t)$.
Then we take the total differential of $\mathcal{L}^*$ w.r.t. $dt$:
\begin{align*}
    d\mathcal{L}^*(t;dt)=&df^*(t)+d\lambda^*(t)^\top g^*(t)+\lambda^*(t)^\top d g^*(t)+d\nu^*(t)^\top h^*(t)+\nu^*(t)^\top dh^*(t)\\
    =& (D_tf^* + D_x f^* D_tx^*)dt\\
    & + dt^\top D_t(\lambda^*)^\top g^* + (\lambda^*)^\top \Big[ D_tg^* +D_xg^* D_x x^*\Big]dt \\
    & + dt^\top D_t(\nu^*)^\top h^* + (\nu^*)^\top \Big[ D_t h^* +D_xh^* D_t x^*(t)\Big]dt \\
    =& D_tf^* dt + \Big[D_xf^* + \lambda^*(t)^\top D_xg^*+\nu^*(t)^\top D_xh^* \Big]D_tx^*dt\\
    &+ dt^\top D_t(\lambda^*)^\top g^*+ dt^\top D_t\nu^*(t)^\top h^* + (\lambda^*)^\top D_tg^*dt + (\nu^*)^\top D_th^*dt\\
\end{align*}
By Lemma \ref{lem1}, at $(t,x^*(t))$, we have the following stationarity, primal feasibility, and complementary slackness conditions:
\[D_xL(x^*,\lambda^*,\nu^*;t)= D_xf^*+(\lambda^*)^\top D_x g^*(t) + (\nu^*)^\top  D_xh^*(t)= \mathbf{0}^\top,\]
\[h(t,x^*)=AT_tx^* -c=0,\]
\[(\lambda^*)^\top g(t,x^*)=(\lambda^*)^\top (GT_tx^* -b)=0.\]
Firstly, since $h(t,x^*)=0$, $dt^\top D_t\nu^*(t)^\top h^*=0$. Identically denote the active set of inequality constraints as $\mathcal{A}:=\{i\in \{1,\dots,2p\}: [GT_tx-b]_i=0\}$ as in Theorem \ref{theo2}. For any $i\in \mathcal{A}$, $g_i^* =0$. On the other hand, for any $i\in \mathcal{A}^c$, suppose it corresponds to the inactive lower bound on $x_j$ for a $j\in \{1,\dots,p\}$, then we have $l_j<t_jx^*_j(t)$. Since $tx^*(t)$ is a continuous function of $t$, there exists $\epsilon >0$ such that for all $\tau \in B_\epsilon (t)$, we have $l_j < \tau_j x_j^*(\tau)$ and $\lambda^*(\tau)=0$ so that the $ith$ row of $D_t(\lambda^*)$ is zero, i.e. $[D_t \lambda^*]_{i,1:p}=0$, which also follows directly from part 2 of Theorem \ref{theo2}. Hence, grouping both cases, we have that $dt^\top D_t(\lambda^*)^\top g^*=0$.\\
By substituting the three equalities into $d\mathcal{L}^*(t;dt)$, we have that
\[d\mathcal{L}^*(t;dt) = D_t\mathcal{L}^*dt= \Big[D_tf^*+(\lambda^*)^\top D_tg^*+(\nu^*)^\top D_th^*\Big]dt\]
where $\nabla_t\mathcal{L}^*=(D_t\mathcal{L}^*)^\top$ which is the same result partial differential result by the Envelop Theorem.\\
Direct computation gives the following closed-form expressions:
\[D_tf^* = 2t^\top T_{x^*}(\Sigma -\delta I)T_{x^*}-\mu^\top T_x,\]
\[D_tg^* =G T_{x^*},\]
\[D_th^* = AT_{x^*},\]
so that gradient of the value function is given by
\begin{align*}
        \nabla_t f^* = &(D_t\mathcal{L}^*)^\top= \Big[D_tf + (\lambda^*)^\top D_t g^* +(\nu^*)^\top  D_t h^* \Big]^\top\\
    = & [2t^\top T_x(\Sigma -\delta I)T_x -\mu^\top T_x +(\lambda^*)^\top GT_x +(\nu^*)^\top AT_x]^\top
    \end{align*}
\end{proof}

\subsection{Proof of Theorem \ref{theo4}}\label{apex1:proof_theo4}

\begin{proof}
    Given $t\in \mathcal{T}$, by the stationarity condition in Lemma \ref{lem1}, we have that
    \[2t_i\sum_{j=1}^p t_j\sigma_{ij}x_j^* + 2\delta (1-t_i^2)x_i^*-t_i\mu_i+t_i(\lambda^*_{p+i}-\lambda^*_i)+ \nu^* t_i=0\]
    If $t_i=0$, then $x_i^*(t)=0$, so that $[\nabla_tf^*]_i=0$.\\
    Suppose that $t_i\in (0,1]$, then divide the KKT stationarity condition by $t_i$ and substitute into the expression of $[\nabla_tf^*]_i$ and we have
    \[[\nabla_t f^*]_i = -2\delta (x_i^{*})^2[t_i+\frac{1-t_i^2}{t_i}]=-2\frac{\delta}{t_i} (x_i^{*})^2\leq 0\]
    where equality holds only when $x_i^*=0$. This is true for all $i\in \{1,\dots , p\}$. Hence, for any $t\in \mathcal{T}$, the gradient components of $\nabla_tf^*$ are all non-positive.\\
    By Theorem \ref{theo2}, $x^*$, $\lambda^*$, $\nu^*$ are all continuous functions of $t$. For any convergent sequence $\{t^{(n)}\}_n$ such that $\{t^{(n)}_i\}_n$ converges to $0$, we have $\underset{n\rightarrow \infty}{\lim} x^*_i(t) = 0$ by Lemma \ref{lem1} part 1. Use the Stationarity condition again to get
    \[\frac{x_i^*}{t_i} = \frac{\mu_i-(\lambda_{p+i}^*-\lambda_i^*)-\nu^* -2\sum_{j\neq i}t_j\sigma_{ij}x_j^*}{2t_i^2\sigma_{ii} + 2\delta (1-t_i^2)},\]
    and denote the corresponding sequence as $\{(\frac{x_i^*}{t_i})^{(n)}\}_n$, then we have
    \[\underset{n\rightarrow \infty}{\lim}(\frac{x^*_i}{t_i})^{(n)}\leq \frac{C}{2\delta}\]
    where $C:=\max_n \{\Big|(\mu_i-(\lambda_{p+i}^*-\lambda_i^*)-\nu^*-2\sum_{j\neq i}t_j\sigma_{ij}x_j^*)^{(n)}\Big|\}_n<\infty$ and $0<\delta <\min\{\text{diag}(\Sigma)\}$. $C$ is well-defined by the Extreme Value Theorem since $\lambda^*_{p+1}$, $\lambda_i^*$, $\nu^*$, $x_j^*$ ($j\neq i$) are continuous on the compact domain $\mathcal{T}=\{t\in[0,1]^p: \|t\|_1 =k\}$ by Theorem \ref{theo2}\\
    As such, we have
    \begin{align*}
        \underset{n\rightarrow \infty}{\lim} \Big|[\nabla_{t^{(n)}}f^*]_i\Big| &=2\delta \lim_{n\rightarrow \infty}\Big|\frac{x_i^*(t^{(n)})}{t^{(n)}_i}\Big|\cdot \Big| x_i^*(t^{(n)})\Big|\\
        &\leq 2\delta \lim_{n\rightarrow \infty}\frac{C}{2\delta}\cdot\Big|x_i^*(t^{(n)})\Big|=0
    \end{align*}
    This concludes the proof.
\end{proof}

\subsection{Proof of Lemma \ref{lem2}} \label{apex1:proof_lem2}

\begin{proof}
\begin{enumerate}
    \item By the definition of the subset selection space, $f^*:\mathcal{P}\rightarrow \mathbb{R}$ assigns a unique real value to each subset $s\in \mathcal{P}$. Hence, for any $a,b,c \in \mathcal{P}$ such that $aR b$, and $bR c$, we have $f^*(a)\leq f^*(b)\leq f^*(c)$ so that $aR c$ and $R$ is transitive. Also, for any $s,r\in \mathcal{P}$, we have either $f^*(s)\leq f^*(r)$ or $f^*(r)\leq f^*(s)$, so that either $rR s$ or $sR r$. Hence, $R$ is strongly connected, which implies reflexivity. Hence, $R$ is a total preorder on $\mathcal{P}$.\\
    \item Since $\forall a,b\in \mathcal{P}$, we have $a\neq b \implies f^*(a)\neq f^*(b)$, then $aRb$ and $bR a$ implies $f^*(a)=f^*(b)$, so that $a=b$ and $R$ is antisymmetric. Hence, $R$ is a total order and $(\mathcal{P},R)$ is totally ordered.
\end{enumerate}
\end{proof}

\section{Figures And Tables}
\subsection{Illustration of FW Descent Scheme over Ill-conditioned Matrices}\label{apex2:FWevo}
\begin{figure}[H]
    \centering
    \includegraphics[width=1\linewidth]{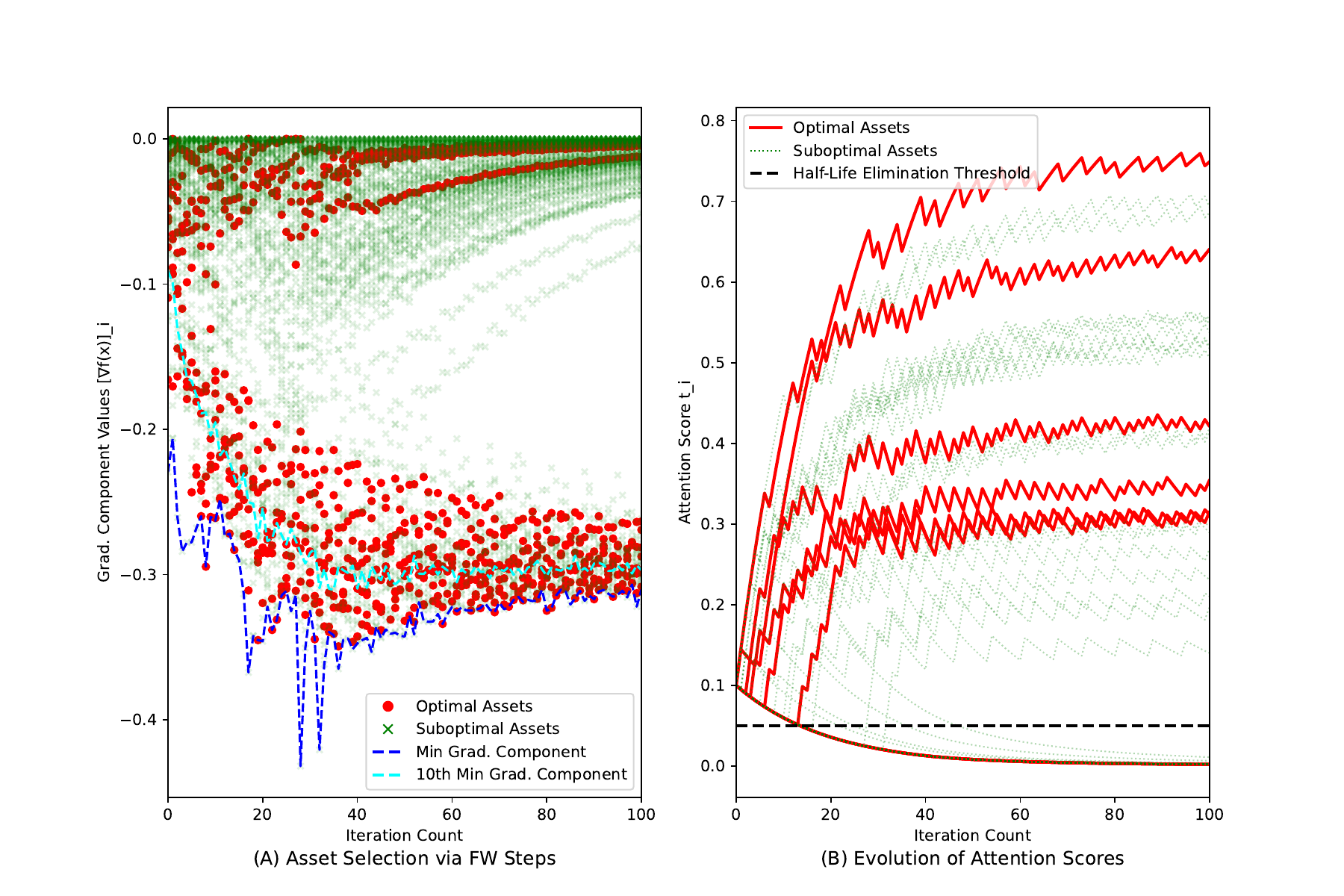}
    \caption{Attention Score Evolution via FW Steps}
    \vspace{2pt}
    {\raggedright \footnotesize \textbf{Note:} The Figure uses synthetic spiked covariance matrix with the following problem parameters: $p=100$, $k=10$, $l=-1$, $u=1$, $\kappa =1e3$, sample $1$ in seed $42$.\par}
    \label{fig:fw_kappa1e3}
\end{figure}

\begin{figure}[H]
    \centering
    \includegraphics[width=1\linewidth]{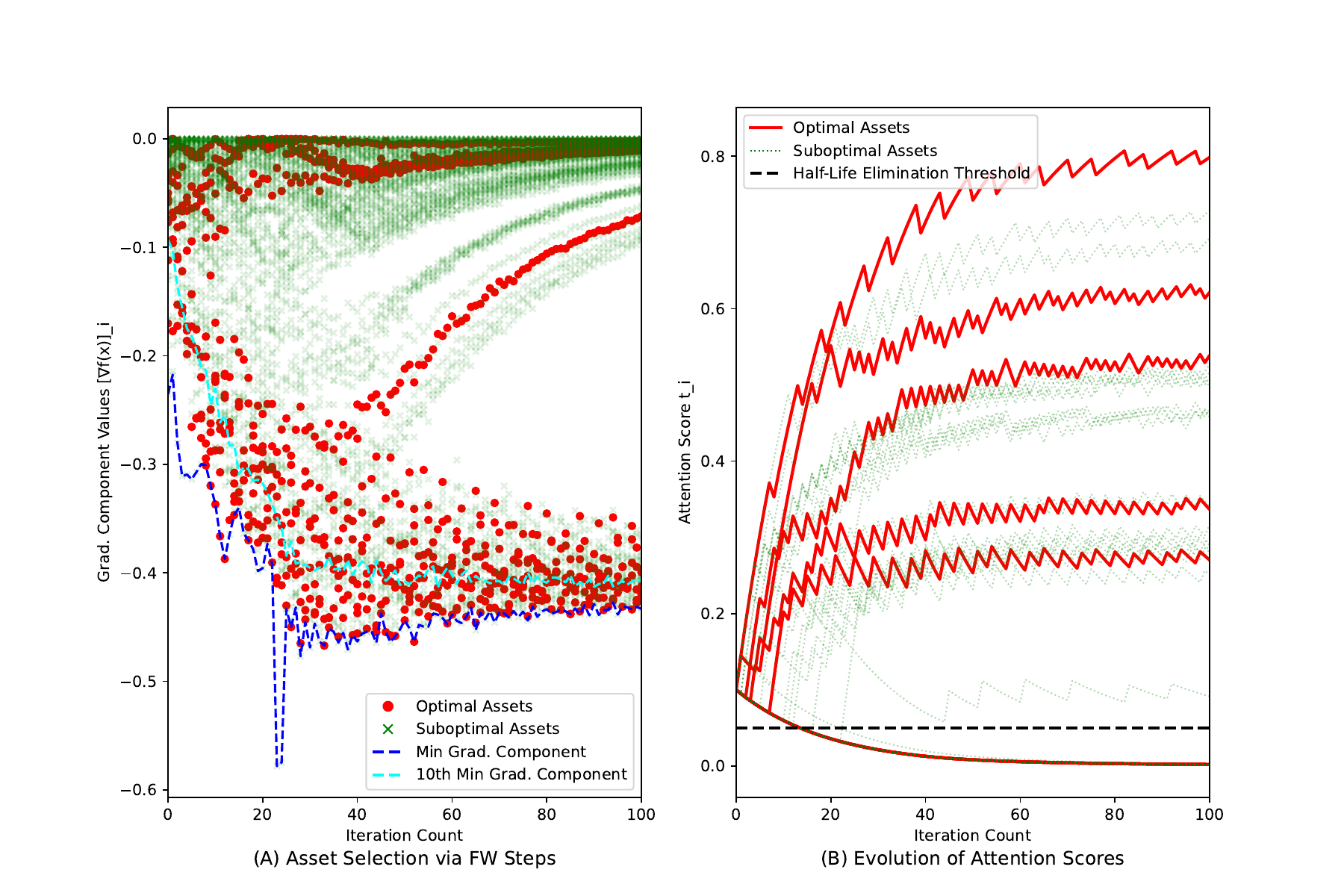}
    \caption{Attention Score Evolution via FW Steps}
    \vspace{2pt}
    {\raggedright \footnotesize \textbf{Note:} The Figure uses synthetic spiked covariance matrix with the following problem parameters: $p=100$, $k=10$, $l=-1$, $u=1$, $\kappa =1e6$, sample $1$ in seed $42$.\par}
    \label{fig:fw_kappa1e6}
\end{figure}
\newpage
\subsection{FWF Efficiency}\label{apex2:fwf}
\begin{figure}[h]
    \centering
    \includegraphics[width=0.8\linewidth]{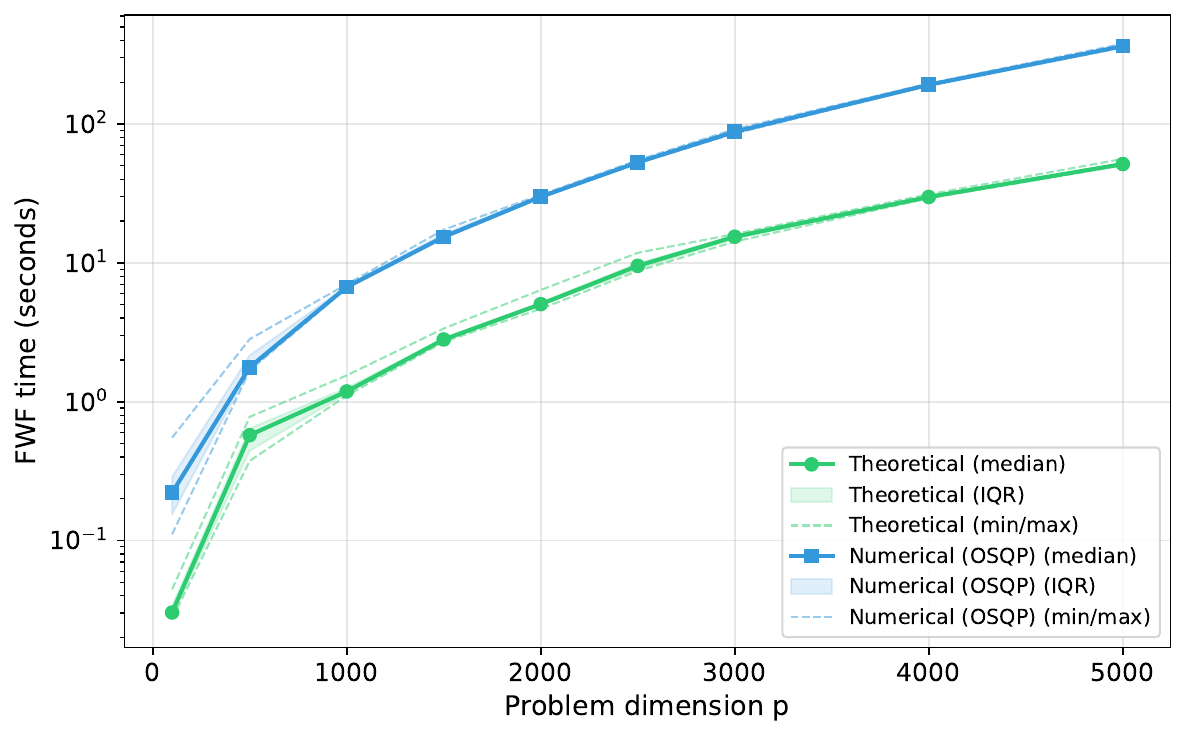}
    \caption{FWF Efficiency: Closed-form vs Envelope Gradient (10 samples per $p$)}
    \label{fig:FWF_benchmark}
\end{figure}
\FloatBarrier
\newpage
\subsection{Additional Performance Tables}\label{apex2:GenPerCom}

\begin{table}[h]
\centering
\footnotesize
\setlength{\tabcolsep}{3pt}
\renewcommand{\arraystretch}{1.3}
\caption{Performance Comparison: $\kappa=1e1$, $[l,u]=[-0.3,1]$}
\begin{adjustbox}{max width = \textwidth}
\begin{tabular}{ccrrrrrrrrrrrr}
\toprule
 &  & \multicolumn{3}{c}{Time (60s)} & \multicolumn{3}{c}{Time (180s)} & \multicolumn{3}{c}{Time (300s)} & \multicolumn{3}{c}{Time (600s)} \\
p & k & DASH & Gurobi & Win\% & DASH & Gurobi & Win\% & DASH & Gurobi & Win\% & DASH & Gurobi & Win\% \\
\midrule
\multirow{3}{*}{100} & 3 (3.0\%) & -0.0082 & -0.0082 & 90.0\% & -0.0082 & -0.0082 & 90.0\% & -0.0082 & -0.0082 & 90.0\% & -0.0082 & -0.0082 & 90.0\% \\
 & 5 (5.0\%) & -0.0099 & -0.0099 & 100.0\% & -0.0099 & -0.0099 & 100.0\% & -0.0099 & -0.0099 & 100.0\% & -0.0099 & -0.0099 & 100.0\% \\
 & 10 (10.0\%) & -0.0127 & -0.0127 & 90.0\% & -0.0127 & -0.0127 & 90.0\% & -0.0127 & -0.0127 & 90.0\% & -0.0127 & -0.0127 & 90.0\% \\
\hline
\multirow{3}{*}{500} & 5 (1.0\%) & -0.0103 & -0.0103 & 90.0\% & -0.0103 & -0.0103 & 90.0\% & -0.0103 & -0.0103 & 90.0\% & -0.0103 & -0.0103 & 90.0\% \\
 & 25 (5.0\%) & -0.0234 & -0.0196 & 70.0\% & -0.0234 & -0.0196 & 50.0\% & -0.0234 & -0.0196 & 40.0\% & -0.0234 & -0.0196 & 40.0\% \\
 & 50 (10.0\%) & -0.0366 & -0.0335 & 70.0\% & -0.0366 & -0.0335 & 70.0\% & -0.0366 & -0.0335 & 70.0\% & -0.0366 & -0.0335 & 70.0\% \\
\hline
\multirow{3}{*}{1500} & 15 (1.0\%) & -0.0181 & -0.0181 & 10.0\% & -0.0181 & -0.0181 & 10.0\% & -0.0181 & -0.0181 & 10.0\% & -0.0181 & -0.0181 & 10.0\% \\
 & 75 (5.0\%) & -0.0572 & -0.0572 & 60.0\% & -0.0572 & -0.0572 & 60.0\% & -0.0572 & -0.0572 & 60.0\% & -0.0572 & -0.0572 & 60.0\% \\
 & 150 (10.0\%) & -0.0959 & -0.0959 & 50.0\% & -0.0959 & -0.0959 & 50.0\% & -0.0959 & -0.0959 & 50.0\% & -0.0959 & -0.0959 & 50.0\% \\
\hline
\multirow{3}{*}{3000} & 30 (1.0\%) & -0.0280 & -0.0277 & 10.0\% & -0.0304 & -0.0277 & 10.0\% & -0.0304 & -0.0277 & 10.0\% & -0.0304 & -0.0277 & 10.0\% \\
 & 150 (5.0\%) & -0.1046 & -0.1046 & 30.0\% & -0.1046 & -0.1046 & 20.0\% & -0.1046 & -0.1046 & 20.0\% & -0.1046 & -0.1046 & 20.0\% \\
 & 300 (10.0\%) & -0.1805 & -0.1805 & 90.0\% & -0.1805 & -0.1805 & 90.0\% & -0.1805 & -0.1805 & 90.0\% & -0.1805 & -0.1805 & 90.0\% \\
\bottomrule
\end{tabular}
\end{adjustbox}
\end{table}

\begin{table}[h!]
\centering
\footnotesize
\setlength{\tabcolsep}{3pt}
\renewcommand{\arraystretch}{1.3}
\caption{Performance Comparison: $\kappa=1e3$, $[l,u]=[-0.3,1]$}
\begin{adjustbox}{max width = \textwidth}
\begin{tabular}{ccrrrrrrrrrrrr}
\toprule
 & & \multicolumn{3}{c}{Time (60s)} & \multicolumn{3}{c}{Time (180s)} & \multicolumn{3}{c}{Time (300s)} & \multicolumn{3}{c}{Time (600s)} \\
p & k & DASH & Gurobi & Win\% & DASH & Gurobi & Win\% & DASH & Gurobi & Win\% & DASH & Gurobi & Win\% \\
\midrule
\multirow{3}{*}{100} & 3 (3.0\%) & -0.0608 & -0.0610 & 60.0\% & -0.0608 & -0.0610 & 60.0\% & -0.0608 & -0.0610 & 60.0\% & -0.0608 & -0.0610 & 60.0\% \\
 & 5 (5.0\%) & -0.0771 & -0.0772 & 70.0\% & -0.0771 & -0.0772 & 70.0\% & -0.0771 & -0.0772 & 70.0\% & -0.0771 & -0.0772 & 70.0\% \\
 & 10 (10.0\%) & -0.1097 & -0.1097 & 80.0\% & -0.1097 & -0.1097 & 80.0\% & -0.1097 & -0.1097 & 80.0\% & -0.1097 & -0.1097 & 80.0\% \\
\hline
\multirow{3}{*}{500} & 5 (1.0\%) & -0.0646 & -0.0648 & 50.0\% & -0.0646 & -0.0648 & 60.0\% & -0.0646 & -0.0648 & 60.0\% & -0.0646 & -0.0648 & 60.0\% \\
 & 25 (5.0\%) & -0.1623 & -0.1627 & 0.0\% & -0.1623 & -0.1627 & 0.0\% & -0.1623 & -0.1627 & 0.0\% & -0.1623 & -0.1627 & 0.0\% \\
 & 50 (10.0\%) & -0.2693 & -0.2695 & 50.0\% & -0.2693 & -0.2695 & 40.0\% & -0.2693 & -0.2695 & 40.0\% & -0.2693 & -0.2695 & 40.0\% \\
\hline
\multirow{3}{*}{1500} & 15 (1.0\%) & -0.1091 & -0.1100 & 0.0\% & -0.1091 & -0.1100 & 0.0\% & -0.1091 & -0.1100 & 0.0\% & -0.1091 & -0.1100 & 0.0\% \\
 & 75 (5.0\%) & -0.3643 & -0.3644 & 40.0\% & -0.3643 & -0.3649 & 20.0\% & -0.3643 & -0.3650 & 20.0\% & -0.3645 & -0.3651 & 20.0\% \\
 & 150 (10.0\%) & -0.6479 & -0.6472 & 80.0\% & -0.6480 & -0.6480 & 40.0\% & -0.6480 & -0.6482 & 40.0\% & -0.6481 & -0.6483 & 40.0\% \\
\hline
\multirow{3}{*}{3000} & 30 (1.0\%) & -0.1705 & -0.1719 & 10.0\% & -0.1706 & -0.1721 & 10.0\% & -0.1706 & -0.1721 & 10.0\% & -0.1706 & -0.1723 & 0.0\% \\
 & 150 (5.0\%) & -0.6496 & -0.6437 & 100.0\% & -0.6505 & -0.6490 & 70.0\% & -0.6507 & -0.6502 & 60.0\% & -0.6508 & -0.6505 & 50.0\% \\
 & 300 (10.0\%) & -1.1759 & -1.1367 & 100.0\% & -1.1813 & -1.1764 & 100.0\% & -1.1815 & -1.1788 & 90.0\% & -1.1817 & -1.1807 & 60.0\% \\
\bottomrule
\end{tabular}
\end{adjustbox}
\end{table}

\section{Preliminary Definitions}
\begin{definition}[preorder and preordered set]
    A binary relation $R$ on a set $X$ is called a preorder if it is reflexive and transive. That is, if it satisfies:
    \begin{enumerate}
        \item Reflexivity: $aRa$ for all $a\in X$, and
        \item Transitivity: if $aRb$ and $bRc$, then $aRc$ for all $a,b,c\in X$.
    \end{enumerate}
    A set that is equipped with a preoder is called a preordered set
\end{definition}
\begin{definition}
    A total preorder on a set $X$ is a preorder in which any two elements are comparable. That is, if it satisfies:
    \begin{enumerate}
        \item Transitivity: For all $x,y,z\in X$, if $xRy$ and $yRz$, then $xRz$
        \item Strong connectedness: For all $x,y\in X$, $xRy$ or $yRx$
        \item Reflexivity: implied by strong connectedness.
    \end{enumerate}
\end{definition}

\end{document}